\documentclass[10pt]{IEEEtran}
\usepackage{amsmath,amssymb}

\usepackage{cite}
\usepackage{graphicx}
\usepackage{psfrag}
\usepackage{epsfig} 
\usepackage{subfigure}
\usepackage{url}
\usepackage{amsmath}
\usepackage{array}
\usepackage{amssymb}
\usepackage{amsfonts}

\begin{document}

\interdisplaylinepenalty=2500
\newtheorem{theorem}{Theorem}
\newtheorem{lemma}{Lemma}
\newtheorem{definition}{Definition}
\newtheorem{proposition}{Proposition}
\newtheorem{corollary}{Corollary}
\newtheorem{example}{Example}
\newtheorem{remark}{Remark}
\newtheorem{note}{Note}
\newcounter{mytempeqncnt}

\title{Maximum-rate, Minimum-Decoding-Complexity STBCs from Clifford Algebras}

\vspace{1.00cm}

\author{Sanjay Karmakar and B. Sundar Rajan, Senior Member, IEEE
\thanks{This work was partly supported by the DRDO-IISc Program on Advanced Research in Mathematical
Engineering and by the Council of Scientific \&
Industrial Research (CSIR), India, through Research Grant (22(0365)/04/EMR-II) to B.S.~Rajan.}
\thanks{The authors are with the Department of Electrical Communication Engineering, Indian Institute of Science, Bangalore, India 560012. email:bsrajan@ece.iisc.ernet.in.}
\thanks{Different parts of the content of this paper appear in the Proc. of IEEE International Symposium on Information on Information Theory (ISIT 2006), Seattle, Washington, July 09-14, 2006.}
}

\markboth{IEEE Transactions on Information Theory,~Vol.~XX, No.~XX,(  Submitted  )}{Karmakar and Rajan: Maximum-rate, MDC STBCs from Clifford Algebras}

\maketitle
\thispagestyle{empty}

\begin{abstract}
It is well known that Space-Time Block Codes (STBCs) from orthogonal designs (ODs) are  single-symbol decodable/symbol-by-symbol decodable (SSD) and are obtainable from unitary matrix representations of Clifford algebras. However, SSD codes are obtainable from designs that are not orthogonal also. Recently, two such classes of SSD codes have been studied: (i) Coordinate Interleaved Orthogonal Designs (CIODs) and (ii) Minimum-Decoding-Complexity (MDC) STBCs from Quasi-ODs (QODs). Codes from ODs, CIODs and MDC-QODs are mutually non-intersecting classes of codes. The class of CIODs have {\it non-unitary weight matrices} when written as a Linear Dispersion Code (LDC) proposed by Hassibi and Hochwald, whereas several known SSD codes including CODs  have {\it unitary weight matrices}. In this paper, we obtain SSD codes with unitary weight matrices (that are not CODs) called Clifford Unitary Weight SSDs (CUW-SSDs) from matrix representations of Clifford algebras. A  main result of this paper is the derivation of an achievable upper bound on the rate of any unitary weight SSD code as $\frac{a}{2^{a-1}}$ for $2^a$ antennas which is larger than that of the CODs which is $\frac{a+1}{2^a}$. It is shown that several known classes of SSD codes are CUW-SSD codes and CUW-SSD codes meet this upper bound. Also, for the codes of this paper conditions on the signal sets which ensure full-diversity and expressions for the coding gain are presented. A large class of SSD codes with non-unitary weight matrices are obtained which include  CIODs as a proper subclass.
\end{abstract}
\section{Introduction and Preliminaries}
\label{sec1}
\pagenumbering{arabic}
We consider a multiple antenna transmission system with $n$  number of transmit antennas and $m$  number of receive antennas. At each time slot $t$, the complex signals, $s_{ti}, \, i=1,2, \cdots ,n$ are transmitted from the $n$ transmit antennas simultaneously. Let $h_{ij}={\alpha}_{ij}{e}^{\mathbf{j}\theta_{ij}}$ denote the path gain from the transmit antenna $i$ to the receive antenna $j$, where $\mathbf{j}=\sqrt{-1}$. Assuming that the path gain are constant over a frame length $n$ (we consider only square designs or square codeword matrices), the received signal $y_{tj}$ at the receive antenna $j$ at time $t$ is given by,
\begin{equation*}
\label{bchmodel}
 y_{tj}={\sum}_{i=1}^{n} s_{ti}h_{ij}+n_{tj},
\end{equation*}
for $j=1,2, \cdots, m, ~~ t=1,2, \cdots, n,$ which in matrix notation is,
\begin{equation*}
\label{mateq1}
Y=SH+N
\end{equation*}
where $Y\in {\mathbb{C}}^{n\times m}$ is the received signal matrix, $S \in {\mathbb{C}}^{n\times n}$ is the transmission matrix (also referred as codeword matrix), $N \in {\mathbb{C}}^{n\times m}$ is the additive noise matrix and $H \in {\mathbb{C}}^{n\times m}$ is the channel matrix, where $\mathbb{C}$ denotes the complex field. The set of all possible codeword matrices $\{s_{ti},~~i,t=1,2,\cdots,n \}$ is the Space-Time Block Code (STBC) used. The entries of $H$ are complex Gaussian with zero mean and unit variance and the entries of $N$ are complex Gaussian with zero mean and variance ${\sigma}^2$. Both are assumed to be temporally and spatially white. We further assume that transmission power constraint is given by $E\left[tr\{SS^{H}\}\right]={n}^2$.

An $n\times n$ linear dispersion STBC \cite{HaH} with $K$ complex variables $x_1,x_2,\cdots,x_K$ is  given by
\begin{equation}
\label{ldceqn}
S  = \sum_{i=1}^{K} (x_{iI}A_{iI}+x_{iQ}A_{iQ})
\end{equation}
\begin{equation}
\label{nopathology}
A_{iQ}  \neq c_i A_{iI}, \mbox{ for some }c_i \in {\mathbb R},~~1 \leq i \leq K
\end{equation}
where ${\bf j}=\sqrt{-1}$ and $x_i= x_{iI}+{\bf j}x_{iQ}$,  $1\leq i \leq K,$  are the $K$ complex variables ($x_{iI}$ and $x_{iQ}$ denoting, respectively,  the in-phase and quadrature components of $x_i$) taking values from a complex signal set ${\cal A}_i$. Then the  number of codewords is $\prod_{i=1}^K {\cal A}_i$. The set of $n\times n$ complex matrices $\{A_{iI}, A_{iQ} \}$, called {\it the weight matrices} define $S$. Notice that in \eqref{nopathology}, it is assumed that the components of the pair $(A_{iI},A_{iQ})$ is not a real scaled version of one another. For otherwise, the code may not be decodable for some signal sets as follows: Suppose $A_{jQ}=cA_{jI}$ for some $j$ and real number $c$. Then in the term $x_{jI}A_{jI}+x_{jQ}A_{jQ} = (x_{jI}+cx_{jQ})A_{jI}$ the real quantity $(x_{jI}+cx_{jQ})$ can turn out to be the same for two different complex signal points leading to the same space time codeword for two different sets of information symbols.

Assuming that perfect channel state information (CSI) is available at the receiver, the maximum likelihood (ML) decision rule  minimizes the metric,
\begin{equation}
\label{mlmetric}
M(S) \triangleq \min_{S} tr({({Y-SH})}^{H}({Y-SH})) = {\parallel {Y}-{SH} \parallel}^2
\end{equation}
where $tr(.)$ denotes the trace of a matrix, $\parallel .\parallel$ denotes the Frobonius norm of the argument and $A^H$ stands for the Hermitian (conjugate transpose) of the matrix $A$. 
It is clear that there are $\prod_{i=1}^{K}|\mathcal{A}_i|$ different codewords and, in general, the ML decoding requires $\prod_{i=1}^{K}|\mathcal{A}_i|$ computations, one for each codeword. If the set of weight matrices  are chosen such that the decoding metric \eqref{mlmetric} could be decomposed into,
\begin{equation*}
M(S)  = \sum_{j=1}^p {f_j(x_{(j-1)q+1},x_{(j-1)q+2},\cdots, x_{(j-1)q+q})}
\end{equation*}
sum of $p$ positive terms, each involving exactly $q$ complex variables only, where $pq=K$,  then the decoding requires $\sum_{j=1}^p \{ \prod_{i=1}^q |\mathcal{A}_{i+(j-1)q} | \}$ ( $ <  \prod_{i=1}^{K}|\mathcal{A}_i|$ as $|\mathcal{A}_i|\geq 2~ \forall~~ i$) computations and the code is called a $q$-symbol decodable code. The case $q=1$ corresponds to \textbf{Single-Symbol Decodable (SSD)} codes that includes the well known Orthogonal Designs (ODs) as a proper subclass, and have been extensively studied  \cite{TJC,Ala,TiH,GaS,KhR1,KhR2,Kha,KhR3,KhR4,KRL1,KRL2,WWX1,WWX2,YGT1,YGT2,YGT3,YGT4,YGT5}. The codes corresponding to $q=2$, are called \textbf{Double-Symbol-Decodable (DSD)} codes. The Quasi-Orthogonal Designs studied in \cite{Jaf,SuX,SuX3,TiH1,ShP1,ShP2,TBH} and \cite{KRL1} are proper subclasses of DSD codes. Codes from Orthogonal designs \cite{TJC,Ala},\cite{GaS} and Quasi-orthogonal designs \cite{Jaf,SuX,SuX3,TiH1,ShP1,ShP2,TBH} and their relationship with Hurwitz-Radon family of matrices \cite{TiH,GaS} and unitary representations of Clifford algebras \cite{TJC}, \cite{TiH}.
\begin{definition}\cite{TJC}
A {\textit{complex orthogonal design}} $ G (x_1, x_2, . . . , x_k )$ (in short $G$)  of size $ p \times n $ is a $ p \times n $ matrix satisfying the following conditions:
\begin{itemize}
\item the entries of $G$ are complex linear combination of  $ x_1 , x_2 , . . . , x_k $ and their complex conjugates $ x_1^* , x_2^* , . . . , x_k^* $ and 
\item (Orthonormality:)
\begin{eqnarray*}
  G^H G=({\vert x_1 \vert}^2 + ...+{\vert x_k \vert}^2)I_n
 \end{eqnarray*}
holds for any complex values for $ x_i , i = 1, 2, . . . , k , $ where $I_n$ is the $ n \times n $ identity matrix and $\vert x \vert$ stands for the magnitude of a complex number $x$.
\end{itemize}
The matrix $ G $ is also said to be a $ [p,n,k] $ complex orthogonal design (COD). If the non-zero entries are the indeterminates $ \pm x_1,\pm x_2,...,\pm x_k $ or their conjugates $ \pm x_1^*,\pm x_2^*,...,\pm x_k^* $ only (not arbitrary complex linear combinations), then $G$ is said to be a {\textit {restricted complex orthogonal design}} (RCOD).
\end{definition}

Notice that the linear STBC $S$ in \eqref{ldceqn} is a complex design which may or may not be a COD. A set of necessary and sufficient conditions for $S$ to be a COD is  \cite{TJC,TiH}
\begin{eqnarray}
\label{A1}
A_{iI}^{H}A_{iI}  =A_{iQ}^{H}A_{iQ}  =I_n,~~~   i=1,2,\cdots,K;
\end{eqnarray}
\begin{subequations}
\label{A2}
\begin{align}
A_{iI}^{H}A_{jQ} + A_{jQ}^{H}A_{iI} & =0 \label{A21} \\
A_{iI}^{H}A_{jI} + A_{jI}^{H}A_{iI} & =0 \label{A22} \\
A_{iQ}^{H}A_{jQ} + A_{jQ}^{H}A_{iQ} & =0 \label{A23}
\end{align}
\end{subequations}
for $1 \leq i \neq j \leq K $, and
\begin{eqnarray}
\label{A3}
A_{iI}^{H}A_{iQ}+ A_{iQ}^{H}A_{iI}  =0,  ~~~   i=1,2,\cdots,K.
\end{eqnarray}
\noindent
STBCs obtained from CODs \cite{TJC},\cite{TiH} are SSD like the well known Alamouti code \cite{Ala}, and  satisfy all the three equations \eqref{A1}, \eqref{A2} and \eqref{A3}. For $S$ to be  SSD  it is not necessary that it satisfies \eqref{A1} and \eqref{A3}; i.e., it is sufficient that it satisfies only  \eqref{A2} - this result was first shown in \cite{KhR1,KhR2,Kha}. Since then, different classes of SSD codes have been studied by several authors, \cite{KhR1,KhR2,Kha,KhR3,KhR4,KRL1,KRL2,WWX1,WWX2,YGT1,YGT2,YGT3,YGT4,YGT5} that are not CODs. To systematically study various possible classes of SSD codes we introduce the following classification:
\begin{enumerate}
\item Linear STBCs satisfying \eqref{A1}, \eqref{A2} and \eqref{A3} are Complex Orthogonal Designs (CODs).
\item Linear STBCs satisfying \eqref{A2} are called {\it SSD codes}; these may or may not satisfy \eqref{A1} and \eqref{A3}.
\item Linear STBCs satisfying \eqref{A1} and \eqref{A2} and not satisfying \eqref{A3} are called Unitary-Weight SSD codes ({\it UW-SSD codes}).
\item Linear STBCs satisfying \eqref{A2} and not satisfying \eqref{A1} are called Non-Unitary weight SSDs ({\it NU-SSD codes}); these may or may not satisfy \eqref{A3} 
\begin{itemize}
\item NU-SSD codes that do not satisfy \eqref{A3} are called Proper-SSD codes ({\it PSSD codes}).
\item NU-SSD codes that satisfy \eqref{A3} are called Non-unitary CODs ({\it NU-CODs}) since these differ from the well known CODs only by the feature that the weight matrices are not unitary.
\end{itemize}
\end{enumerate}
Fig. \ref{fig1p2} shows all these classes of codes along with some more classes of codes discussed in the sequel. The codes discussed in  \cite{KhR1, KhR2, Kha,KhR3,KhR4,KRL1, KRL2}, calling them {\it Coordinate Interleaved Orthogonal Designs, (CIODs)} constitute an example class of NU-SSD codes. The classes of codes  studied in \cite{YGT1}-\cite{YGT5} are UW-SSD codes. The classes of codes studied in \cite{WWX1,WWX2} called Minimum Decoding Complexity codes from Quasi-Orthogonal Designs (MDC-QOD codes) include UW-SSD and NU-SSD codes including PSSD codes. 


The notion of SSD codes have been extended to coding for MIMO-OFDM systems in \cite{DSKR,GoR} and recently, low-decoding complexity codes called 2-group and 4-group decodable codes \cite{KiR1,KiR2,KiR3,RaR} and SSD codes \cite{YiK} in particular are studied for use in cooperative networks as distributed STBCs.

In this paper, we construct several classes of SSD codes from representations (both irreducible and reducible) of Clifford algebras and study their full-diversity, coding gain and rate properties. Specifically, the contributions of this paper are
\begin{itemize}
\item derivation of an upper bound on the rate of any UW-SSD code.
                                                                                
\item construction of a class of UW-SSD codes from matrix representations of real Clifford algebras with rate equal to the upper bound.
                                                                                
\item identification of signal sets which will give full-diversity for the class of codes constructed using Clifford algebras.
                                                                                
\item identification of code and signal set parameters that influence the coding gain of the codes constructed.

\item By using a pair of linear transformations on the weight matrices of any UW-SSD code we obtain a class of NU codes called Transformed Non-Unitary codes (TNU codes). We show that these are indeed SSD codes \eqref{nussdthm} and call them TNU-SSD codes.
                                                                                
\item We identify a set of necessary and sufficient conditions for a TNU-SSD code to be a PSSD code (Theorem \ref{nupssdthm}).
                                                                                
\item We identify the class of linear transformations using which the resulting TNU-SSD codes coincide with CIODs of \cite{KhR1}-\cite{KRL2} (Section \ref{sec4}).

\item Every Clifford algebra is with respect to an underlying quadratic space. Generally one uses Clifford algebras that are with respect to the Euclidean quadratic space (see Appendix -\ref{Append1}). In Appendix \ref{Append2} it is shown that SSD codes can be constructed using  Clifford algebras based on Minkowski spaces. Also, it is proved that if normalized these codes coincide with CUW-SSD codes. 

\end{itemize}

\begin{figure}
\centering
\includegraphics[width=6.0cm,height=6.0cm]{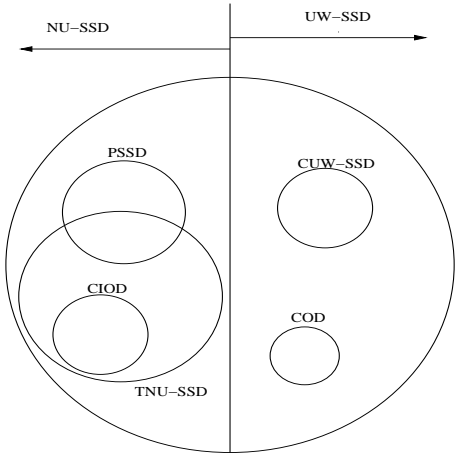}
\caption{Interrelationship of TNU-SSD codes with other classes}
\label{fig1p2}
\end{figure}
                                                                                
Notice that Fig. \ref{fig1p2} shows the class of TNU-SSD codes in relation to CODs, CIODs and the UW-SSD codes of this paper denoted as CUW-SSD codes.
                                                                                
The remaining part of the paper is organized as follows: In Section \ref{sec2} we present a set of sufficient conditions for a linear STBC to be SSD and which are also constructable using representations of real Clifford algebras. The codes satisfying this set of conditions are called Clifford Unitary Weight SSD (CUW-SSD) codes. Section \ref{sec3} presents the construction of CUW-SSD codes and illustrates with examples. Also, a known class of SSD codes is shown to be CUW-SSD codes. It is shown in Section \ref{sec4} that the achievable upper bound on the rate of a UW-SSD code is $\frac{a}{2^{a-1}}.$ Diversity gain and coding gain of CUW-SSD codes are studied in Section \ref{sec5} and compared with those of MDC-SSD codes. Few simulation results are also presented. In Section \ref{sec6}, NUW-SSD codes are discussed- the classes of TNU-SSD codes, PSSD codes and NU-CODs are studied. Section \ref{sec7} shows that the class of CIODs is obtainable as a special case of TNU-SSD codes. Concluding remarks and several directions for further research constitute Section \ref{sec8}. Appendix \ref{Append1} gives a self-contained introduction to quadratic forms, quadratic spaces and different kinds of Clifford algebras. Appendix \ref{Append2} presents SSD codes constructed based on Minkowski Clifford algebras. It is shown that these codes become CUW-SSD codes when normalized. 

\section{Clifford UW-SSD codes}
\label{sec2}
It is well known that SSD codes are closely related to Hurwitz-Radon family of matrices and also Clifford algebras \cite{GaS},\cite{TiH}. In the following sections we obtain a large class of UW-SSD codes using representations of different Clifford Algebras.  In this section, we introduce an important notion called {\it normalizing a linear STBC} which not only simplifies the analysis of the codes but also provides deep insight various aspects of different classes of codes discussed in this paper.   

Towards this end, let 
\begin{equation}
\label{ldcuwssd}
S_U = \sum_{i=1}^{K} (x_{iI}A_{iI}^\prime+x_{iQ}A_{iQ}^\prime)
\end{equation}
be a Unitary Weight code (UW code), i.e., for which all the weight matrices are unitary. We normalize the weight matrices of the code as 
\begin{eqnarray}
\label{normalizessd}
\begin{array}{rl}
A_{iI} & = A_{1I}^{\prime H} A_{iI}^\prime \\
A_{iQ} & = A_{1I}^{\prime H} A_{iQ}^\prime.
\end{array}
\end{eqnarray}
to get {\it the normalized version} of \eqref{ldcuwssd} to be
\begin{equation}
\label{ldcuwnssd}
S_{N} = x_{1I}I_n+x_{1Q}A_{1Q}+\sum_{i=2}^{K} (x_{iI}A_{iI}+x_{iQ}A_{iQ}).
\end{equation}
We call the code $S_{N}$ to be the normalized code of $S_U$. 
\begin{theorem}
\label{ssdnochange}
The code $S_U$ is SSD iff  $S_N$ is SSD. In other words normalization does not affect the SSD property.
\end{theorem}
\begin{proof} For $1 \leq i \neq j \leq K$, all the three equations of \eqref{A2} are satisfied by the weight matrices of $S_U$ iff they are satisfied by the weight matrices of $S_N$ as shown below:

\begin{eqnarray*}
\begin{array}{l}
(i) A_{iI}^{H}A_{jQ} + A_{jQ}^{H}A_{iI}  =0 \\
\Leftrightarrow A_{iI}^{\prime H}A_{1I}^\prime A_{1I}^{\prime H}A_{jQ}^\prime + A_{jQ}^{\prime H} A_{1I}^\prime A_{1I}^{\prime H}A_{iI}  =0 \\
\Leftrightarrow A_{iI}^{\prime H}A_{jQ}^\prime + A_{jQ}^{\prime H}  A_{iI}^\prime  =0 \\
\end{array} 
\end{eqnarray*}
\begin{eqnarray*}
\begin{array}{ll}
(ii) A_{iI}^{H}A_{jI} + A_{jI}^{H}A_{iI}  =0 \\
\Leftrightarrow A_{iI}^{\prime H} A_{1I}^\prime A_{1I}^{\prime H}A_{jI}^\prime + A_{jI}^{\prime H} A_{1I}^\prime A_{1I}^{\prime H}A_{iI}^\prime  =0 \\
\Leftrightarrow A_{iI}^{\prime H}A_{jI}^\prime + A_{jI}^{\prime H} A_{iI}^\prime  =0 \\
\end{array}
\end{eqnarray*}
\begin{eqnarray*}
\begin{array}{ll}
(iii) A_{iQ}^{H}A_{jQ} + A_{jQ}^{H}A_{iQ}  =0 \\
\Leftrightarrow A_{iQ}^{\prime H} A_{1I}^\prime A_{1I}^{\prime H}A_{jQ}^\prime + A_{jQ}^{\prime H} A{1I}^{\prime} A_{1I}^{\prime H}A_{iQ}^\prime  =0 \\
\Leftrightarrow A_{iQ}^{\prime H}A_{jQ}^\prime + A_{jQ}^{\prime H} A_{iQ}^\prime  =0 \\
\end{array}
\end{eqnarray*}
\end{proof}
The following theorem shows  that this normalization does not alter the coding gain also. 
\begin{theorem}
\label{gainnochange}
$S_U$ and $S_{N}$ have the same coding gain.
\end{theorem}
\begin{proof}
Let $\mathbf{DP}(S_U)$ and $\mathbf{DP}(S_{N})$ respectively denote the diversity product of $S_U$ and $S_{N}$. Then  
{\small
\begin{equation}
\label{div_prod}
\mathbf{DP}(S_U)\triangleq\frac{1}{2\sqrt{n}} \min_{S_U \neq \widetilde{S_U}} \Bigg\lvert  \mathrm{det} \Bigg[{\left(S_U-\widetilde{S_U}\right)}^{H} \left(S_U-\widetilde{S_U}\right) \Bigg] \Bigg\rvert   \end{equation}
}
where
\begin{displaymath}
S_U-{\widetilde{S_U}}= \sum_{i=1}^{k} \left(     \triangle x_{iI}A_{iI}^\prime+\triangle x_{iQ}A_{iQ}^\prime\right) \end {displaymath}
Inserting which in  \eqref{div_prod} we get, if $\triangle \boldsymbol{x}= \left( \triangle x_{1},\triangle x_{2}, \cdots \triangle x_{k} \right)$
                                                                                
\begin{eqnarray} 
\label{dpguwssd}
\mathbf{DP}(S_U)=\frac{1}{2\sqrt{n}} \min_{\triangle\mathbf{x}\neq \mathbf{0}} \Bigg\lvert  \mathrm{det} \Bigg[\sum_{i=1}^{k} \Big( \left( {\triangle x_{iI}}^{2}+{\triangle x_{iQ}}^{2} \right)I_{n}  \nonumber \\
{}+\triangle x_{iI}\triangle x_{iQ}\left(A_{iI}^{\prime H}A_{iQ}^\prime+A_{iQ}^{\prime H}A_{iI}^\prime \right). \Big) \Bigg]  \Bigg \rvert  \end{eqnarray}
\noindent
Similarly, for the normalized code, we have                                     
\begin{eqnarray}
\label{dpnuwssd}
\mathbf{DP}(S_{N})  =  \frac{1}{2\sqrt{n}} \min_{\triangle\mathbf{x}\neq \mathbf{0}}  \Bigg\lvert  \mathrm{det} \Bigg[\sum_{i=1}^{k} \Big( \left( {\triangle x_{iI}}^{2}+{\triangle x_{iQ}}^{2} \right)I_{n} \nonumber \\
 {}+\triangle x_{iI}\triangle x_{iQ}\left(A_{iI}^{H}A_{iQ}+A_{iQ}^{H}A_{iI}\right) \Big) \Bigg] \Bigg \rvert  
\end{eqnarray}
\noindent                                                                       Now from the normalization process \eqref{normalizessd} we have, for all $i$, 
\begin{eqnarray*}
\begin{array}{rl}
 A_{iI}^{H}A_{iQ}+A_{iQ}^{H}A_{iI} & = A_{iI}^{\prime H}A_{1I}^\prime A_{1I}^{\prime H}A_{iQ}^\prime+A_{iQ}^{\prime H}A_{1I}^\prime A_{1I}^{\prime H}A_{iI}^\prime \\
        & = A_{iI}^{\prime H}A_{iQ}^\prime+A_{iQ}^{\prime H}A_{iI}^\prime 
\end{array}
\end{eqnarray*}
since $A_{1I}^\prime A_{1I}^{\prime H}=I_{n}$  which implies that the expressions in \eqref{dpguwssd} and \eqref{dpnuwssd} are identical, i.e.,  $\mathbf{DP}(S_U)=\mathbf{DP}(S_{N})$.
\end{proof}

The following theorem identifies a set of sufficient conditions for a UW code to be UW-SSD. In the sequel, we will provide several constructions of  UW-SSD codes using representations of real Clifford algebras satisfying these sufficient conditions.
\begin{theorem}
\label{thm2}
An $n \times n$ UW code described by \eqref{ldcuwssd} and its normalized version given by \eqref{ldcuwnssd} are both UW-SSD code if the weight matrices of the normalized code satisfy the following conditions:
\begin{eqnarray}
\label{suffcondssd}
\begin{array}{rl}
 A_{iI}^{H}  & =-A_{iI} ~~~  2 \leq i \leq K \\
 A_{iI}A_{jI}  & = - A_{jI}A_{iI}, \quad   2\leq i\neq j\leq K \\
A_{1Q}^{H}   &=A_{1Q} \\
 A_{iQ}   & =A_{1Q}A_{iI}, ~~~~  2 \leq i\leq K \\
A_{1Q}A_{jI} & =A_{jI}A_{1Q} ~~~~ 1 \leq j \leq K \\
\end{array}
\end{eqnarray}
\end{theorem}

\begin{proof}
The proof is by direct verification of \eqref{A2} for the weight matrices. \\
{\bf Proof for the normalized code:}
\begin{eqnarray*}
\begin{array}{rl}
 A_{iI}^{H}A_{jQ}+A_{jQ}^{H}A_{iI} & = A_{iI}^{H}A_{1Q}A_{jI}+A_{jI}^{H}A_{1Q}^{H}A_{iI}  \\
&  = -\left[A_{iI}A_{1Q}A_{jI}+A_{jI}A_{1Q}A_{iI} \right]\\
&  = -(A_{iI}A_{jI}+A_{jI}A_{iI})A_{1Q}\\
&  = 0(A_{1Q})=0.
\end{array}
\end{eqnarray*}
This shows that \eqref{A21} is satisfied for the normalized code. Next, we show that \eqref{A22} is also satisfied:
\begin{eqnarray*}
\begin{array}{rl}
A_{iI}^HA_{jI}+A_{jI}^HA_{iI} & = -(A_{iI}A_{jI}+A_{jI}A_{iI})=0.\\
\end{array}
\end{eqnarray*}
To prove \eqref{A23}:
\begin{eqnarray*}
\begin{array}{rl}
A_{iQ}^HA_{jQ}+A_{jQ}^HA_{iQ} & = A_{iI}^HA_{1Q}^HA_{1Q}A_{jI}+A_{jI}^HA_{1Q}^HA_{1Q}A_{iI}\\
& =A_{iI}^HA_{jI}+A_{jI}^HA_{iI} \\
& = -(A_{iI}A_{jI}+A_{jI}A_{iI}) =0.
\end{array}
\end{eqnarray*}
This shows that the normalized code is UW-SSD. The proof for the unnormalized code follows from Theorem \ref{ssdnochange}. 

\end{proof}
\begin{definition}
A UW-SSD code satisfying the conditions of \eqref{suffcondssd} is defined to be a Clifford Unitary Weight SSD (CUW-SSD) codes.
\end{definition}

The name in the above definition is due to the fact that such codes are constructable using matrix representations of real Clifford algebras which is shown in the following section.
\section{Construction of CUW-SSD codes}
\label{sec3}
Our construction of new classes of both UW-SSD codes and Non-Unitary SSD codes will  make use of the matrix representations (both reducible and irreducible) of different real Clifford algebras. Moreover, in Section \ref{sec5} an upper bound on the rate of CUW-SSD codes is obtained making extensive use of properties of representations real Clifford algebras.  Hence, in Appendix \ref{Append1} we give a brief and self-contained introduction to quadratic forms, quadratic spaces and the associated Clifford algebras. It is assumed that the reader is familiar with basic ideas concerning algebras \cite{Jac}. Every Clifford algebra is based on a quadratic space. Generally Clifford algebras based on Euclidean quadratic spaces are used in the STBC literature as well as throughout this paper except in Appendix \ref{Append2} where using Clifford algebras based on Minkowski quadratic spaces we construct UW-SSD codes and call them MCUW-SSD codes. It is also shown that when normalized these codes coincide with CUW-SSD codes.                                                                       
\subsection{CUW-SSD codes from Euclidean Clifford algebras}
In this subsection we obtain CUW-SSD codes from Euclidean Clifford algebras (see Appendix-I) and in Appendix-II we construct UW-SSD codes from Minkowski Clifford algebras. 
\begin{definition}
The Euclidean Clifford algebra, denoted by $CA_L$, which was described in Appendix-I in terms of an appropriate quadratic form can also be defined as the algebra over the real field $\mathbb R$ generated by $L$ objects $\gamma_k, ~~~k=1,2,\cdots,L$ which are anti-commuting
$$\gamma_k\gamma_j = -\gamma_j \gamma_k, ~~~~ \forall k\neq j$$
and squaring to $-1$
$$\gamma_k^2 = -1 ~~~~ \forall k=1,2,\cdots,L.$$
\end{definition} 
The basis of $CA_L$ is
{\small
$$B_L=\{1\} \bigcup \{\gamma_k \}_{k=1}^{L} \bigcup_{m=2}^L\{\prod_{i=1}^m \gamma_{k_i}| i \leq k_i < k_{i+1} \leq L  \}.$$
}
Note that the number of basis elements is the number of non-ordered combinations of $L$ objects which is $2^L$.
                                                                                
A matrix representation of an algebra is completely specified by the representation of its basis, which in turn is completely specified by a representation of its generators. For a Clifford algebra, we are thus interested in matrix representation of the generators $\gamma_k$'s. In $N$-dimensional representation 1 is represented by $I_N$, the $N\times N$ identity matrix and the generators are anti-commuting matrices that square to $-I_N$. In the following sections, we will use the fact that a double cover of the basis of a Clifford algebra
\begin{equation}
\label{doublecover}
G_L = B_L \bigcup \{ -b | b \in B_L\}
\end{equation}
is a finite group \cite{TiH}.
                                                                                
\begin{lemma}
\label{lemma2}
We can have $2a-1$ Hurwitz-Radon matrices in $N=2^a$ dimension along with a non-identity Hermitian matrix which commutes with all these $2a-1$ matrices.
\end{lemma}
\begin{proof} Let 
\begin{equation}
\label{paulimatrices}
\sigma_1 =\left[ \begin{array}{rr}
0 & 1 \\
-1 & 0
\end{array}
\right],
\sigma_2 =\left[ \begin{array}{rr}
0 & j \\
j & 0
\end{array}
\right], 
\sigma_3 =\left[ \begin{array}{rr}
1 & 0 \\
0 & -1
\end{array}
\right]
\end{equation}
and $ ~~~~~~~~$
$A^{\otimes^{m}} = \underbrace{A\otimes A\otimes A \cdots \otimes A }_{m~~times  } $. \\
From \cite{TiH} we know that the representation of the generators of $CA_{2a+1}$ is given by
\begin{equation}
\label{repmatrices}
\begin{array}{rl}
R(\gamma_2) &= I_2^{\otimes^{a-1}}  \bigotimes \sigma_1 \\
                                                                                
R(\gamma_3) &= I_2^{\otimes^{a-1}}  \bigotimes \sigma_2 \\                      . & . \\
. & . \\
. & . \\
R(\gamma_{2k}) &= I_2^{\otimes^{a-k}} \bigotimes  \sigma_1 \bigotimes \sigma_3^{\otimes^{k-1}} \\
R(\gamma_{2k+1}) &= I_2^{\otimes^{a-k}} \bigotimes \sigma_2 \bigotimes \sigma_3^{\otimes^{k-1}} \\
. & . \\
. & . \\
. & . \\
R(\gamma_{2a}) &= \sigma_1 \bigotimes \sigma_3^{\otimes^{a-1}} \\            R(\gamma_{2a+1}) &= \sigma_2 \bigotimes \sigma_3^{\otimes^{a-1}} \\
R(\gamma_1) &=\pm j \sigma_3^{\otimes^{a}}.
\end{array}
\end{equation}
From the above list of representation matrices we take the first $(2a-1)$ of them, i.e.,
\begin{equation}
\label{requiredHR}
\{ R(\gamma_2), R(\gamma_3), \cdots,  R(\gamma_{2a}) \}
\end{equation}
as our required set of H-R matrices and
\begin{equation*}
\label{requiredHer}
R^\prime(\gamma_1) =jR(\gamma_{2a+1})R(\gamma_1) =j\sigma_1 \otimes I_2^{\otimes^{a-1}} 
\end{equation*}
to be the required Hermitian matrix. 

Using the relation $\sigma_1 \sigma_2 = j \sigma_3$ and the following properties of the tensor products of matrices $A,B,C$ and $D$
\begin{eqnarray*}
\begin{array}{rl}
(A \bigotimes B)^H & = A^H \bigotimes B^H \\
(A \bigotimes B)(C \bigotimes D) & = AC \bigotimes BD
\end{array}
\end{eqnarray*}
it can be easily checked that $R^\prime(\gamma_1)$ commutes with all the $(2a-1)$  matrices of \eqref{requiredHR}.
\end{proof}
Now, we are ready to construct the CUW-SSD codes. Theorem~\ref{thm2} and Lemma~\ref{lemma2} suggests an elegant method of constructing rate $\frac{a}{2^{a-1}}$ UW-SSD codes. Now, we describe this construction in the following theorem followed  by  illustrative  examples. 
\begin{theorem}
\label{cuwssd}
Consider the following $2^a\times 2^a$  weight matrices 
\begin{equation}
\label{wtmatrices}
\begin{array}{rl}
A_{1I} &= I_n \\
A_{iI} &= R(\gamma_i), ~~ 2 \leq i \leq 2a \\
\mbox{  and       } A_{iQ} & =A_{1Q}A_{iI}, ~~~~~ 2 \leq i \leq 2a \\

\mbox{ where      } A_{1Q} &= j \sigma_1 \bigotimes I_2^{\otimes^{a-1}}
\end{array}
\end{equation}
and $\sigma_1, \sigma_2$ and $\sigma_3$ are given by \eqref{paulimatrices}. 
 With these weight matrices the resulting  $2^a \times 2^a$ code $S(x_1,x_2,\cdots,x_{2a})$ given by \eqref{cuwssdeq} at the top of the next page,  where
\begin{figure*}
\begin{equation}
\label{cuwssdeq}
\begin{array}{l}
 \sigma_{x_1} \bigotimes I_2^{\otimes^{a-1}} +\rho_{x_{2a}} \bigotimes \sigma_3 ^{\otimes^{a-1}}
  + \sum_{i=1}^{a-1}
\left[
\sigma_{x_{2i}} \bigotimes I_2^{\otimes^{a-i-1}}  \bigotimes \sigma_1 \bigotimes \sigma_3^{\otimes^{i-1}}+\sigma_{x_{2i+1}} \bigotimes I_2^{\otimes^{a-i-1}} \bigotimes \sigma_2 \bigotimes \sigma_3^{\otimes^{i-1}}
\right]  
\end{array}  
\end{equation} \hrule
\end{figure*}  
\[\begin{array}{rl}
x_i &= x_{iI}+jx_{iQ} \\
                                                                                
\sigma_{x_i} &= \left[
\begin{array}{rr}
x_{iI} & jx_{iQ} \\
-jx_{iQ} & x_{iI}
\end{array}
\right]  \mbox{  and }\\
                                                                                
\rho_{x_i} &= \left[
\begin{array}{rr}
-jx_{iQ} & jx_{iI} \\
-x_{iI} & -jx_{iQ}
\end{array}
\right]                                            
\end{array}
\]
is a CUW-SSD code in $2a$ complex variables with rate ($\frac{a}{2^{a-1}}$).
\end{theorem}
\begin{proof}
From the representation matrices of Lemma~\ref{lemma2} and by the construction of weight matrices it is easily checked by direct verification that all the sufficient conditions of Theorem \ref{thm2} given by \eqref{suffcondssd} for an UW-SSD are satisfied.
\end{proof}
\begin{remark}
In Theorem \ref{cuwssd} the first $(2a-1)$ matrices of the list \eqref{repmatrices} have been set equal to the $(2a-1)$ matrices $A_{iI},~~i=2,\cdots,2a,$ and the product of the remaining two matrices of the list have been set equal to $R_{1Q}.$  It can be verified that the theorem holds if we set any $(2a-1)$ matrices of the list \eqref{repmatrices} to be $A_{iI},~~i=2,\cdots,2a$ and the product of the remaining two to be $A_{1Q}.$ 
\end{remark}
\begin{definition}
The $2^a\times 2^a$ STBCs given by \eqref{cuwssdeq} are defined to be a $2^a-$Clifford Unitary Weight SSD (CUW-SSD) code.
\end{definition}

The $2-$CUW-SSD code is
\[
\begin{array}{c}
S(x_1,x_2)=\sigma_{x_1}+\rho_{x_2} = \left[
\begin{array}{rr}
x_{1I}-jx_{2Q}  & x_{2I}+jx_{1Q} \\
-x_{2I}-jx_{1Q} & x_{1I}-jx_{2Q}
\end{array}
\right]
\end{array}
\]
and the $4-$CUW-SSD code is 
\begin{eqnarray*}
\begin{array}{l}
S(x_1,x_2,x_3,x_4) \\
=\sigma_{x_1}\bigotimes I_2 + \rho_{x_1} \bigotimes \sigma_3 +\sigma_{x_2} \bigotimes \sigma_1 +\sigma_{x_3}\bigotimes \sigma_2
\end{array}
\end{eqnarray*}
which is 
{\small
\begin{eqnarray*}  
\left[
\begin{array}{rrrr}
x_{1I}-jx_{4Q}  & x_{2I}+jx_{3Q} & x_{4I}+jx_{1Q}&-x_{3Q}+jx_{2Q}  \\
-x_{2I}-jx_{3I} & x_{1I}-jx_{4Q} & -x_{3Q}-jx_{2Q} & -x_{4I}+jx_{1Q} \\
                                                                                
-x_{4I}-jx_{1Q} & x_{3Q}-jx_{2Q} & x_{1I}-jx_{4Q} & x_{2I}+jx_{3I} \\
x_{3I}+jx_{2Q} & x_{4I}-jx_{1Q} & -x_{2I}+jx_{3I} & x_{1I}+jx_{4Q}
\end{array}
\right].
\end{eqnarray*}
}

\subsection{YGT codes are CUW-SSD codes}
   In \cite{YGT2} and \cite{YGT4} Yuen, Guan and Tjhung have constructed a class of MDC-QOD codes, which are SSD (with Unitary weight matrices) from Orthogonal designs.  We call these codes YGT codes and show in this subsection that these codes form a proper subclass of CUW-SSD codes.

For constructing a $n\times n$ MDC-QOD code where $n=2^a$, YGT codes begin with  an $\frac{n}{2}\times \frac{n}{2}$ orthogonal design,
\begin{equation}
\label{oygt1}
\mathbf{S}^{OD}_{\frac{n}{2}}={\sum}_{u=1}^{K}x_{uI}\underline{A_u}+jx_{uQ}\underline{B_u}
\end{equation}
and construct the $n\times n$ weight matrices of the MDC-QOD code,
\begin{equation}
\label{oygt2}
\mathbf{S}^{MDC-QOD}_{n}={\sum}_{u=1}^{2K}x_{uI}A_u+jx_{uQ}B_u
\end{equation}
in the following way,
\begin{eqnarray*}
\begin{array}{cc}
A_u=\left[\begin{array}{cc}
\underline{A_u} & 0\\
0 & \underline{A_u} \end{array}\right] & A_{u+K}=\left[\begin{array}{cc}
j\underline{B_u} & 0\\
0 & j\underline{B_u} \end{array}\right]   \\

B_u=\left[\begin{array}{cc}
0 & j\underline{A_u} \\
j\underline{A_u} & 0  \end{array}\right] &  B_{u+K}=\left[\begin{array}{cc}
0 & \underline{B_u} \\
\underline{B_u} & 0 \end{array}\right].  
\end{array}
\end{eqnarray*}
Note that in writing the expression for the linear dispersion codes in  \eqref{oygt1} and \eqref{oygt2} the $j$ has not been included in the corresponding weight matrices. But in our construction we have absorbed the $j$ in the corresponding weight matrices. To facilitate comparison,  we describe the construction procedure in a different way taking $j$ into the corresponding weight matrices. For constructing a $n\times n$ MDC-QOD code where $n=2^a$, we take an $\frac{n}{2}\times \frac{n}{2}$ orthogonal design,
\begin{equation*}
\label{mygt1}
\mathbf{S}^{OD}_{\frac{n}{2}}={\sum}_{u=1}^{K}x_{uI}\underline{A^{\prime}_u}+x_{uQ}\underline{B^{\prime}_u}
\end{equation*}
(here $\underline{A^{\prime}_u}=\underline{A_u}$ and $\underline{B^{\prime}_u}=j\underline{B_u}$) and construct the $n\times n$ weight matrices of the MDC-QOD code,
\begin{equation*}
\label{mygt2}
\mathbf{S}^{MDC-QOD}_{n}={\sum}_{u=1}^{2K}x_{uI}A^{\prime}_u+x_{uQ}B^{\prime}_u
\end{equation*}
\noindent
(here $A^{\prime}_u=A_u$ and $B^{\prime}_u=jB_u$) in the following way,
\begin{eqnarray*}
\begin{array}{cc}
A^{\prime}_u=\left[\begin{array}{cc}
\underline{A^{\prime}_u} & 0\\
0 & \underline{A^{\prime}_u} \end{array}\right] & A^{\prime}_{u+K}=\left[\begin{array}{cc}
\underline{B^{\prime}_u} & 0\\
0 & \underline{B^{\prime}_u} \end{array}\right]   \\
                                                                                
B^{\prime}_u=\left[\begin{array}{cc}
0 & -\underline{A^{\prime}_u} \\
-\underline{A^{\prime}_u} & 0  \end{array}\right] &  B^{\prime}_{u+K}=\left[\begin{array}{cc}
0 & \underline{B^{\prime}_u} \\
\underline{B^{\prime}_u} & 0 \end{array}\right]  \end{array}
\end{eqnarray*}
Note that these weight matrices have the following structure,
\begin{eqnarray*}
A^{\prime}_1=I_{n\times n}, \quad {\{A^{\prime}_u\}}_{u=2}^{2a} \, \textrm{ is an HR family} \\
B^{\prime}_1=j{\sigma}_2\otimes I_{\frac{n}{2}\times \frac{n}{2}}, \quad B^{\prime}_u=\pm B^{\prime}_1 A^{\prime}_u \textrm{  for  } 1 \leq u \leq 2a.
\end{eqnarray*}
Note that $B^{\prime}_1$ is a unitary Hermitian matrix that commutes with all $A^{\prime}_u$ for $2 \leq u \leq 2a$. Hence the YGT codes satisfy all the conditions of \eqref{suffcondssd} and has the following two special features which have been obtained without the use of representations of Clifford algebras.
\begin{itemize}
\item The set, ${\{A^{\prime}_u\}}_{u=2}^{2a}$ is constructed in a particular way.
\item $B^{\prime}_1$ is a special matrix satisfying all the constraints in \eqref{suffcondssd}.
\end{itemize}
If we choose a different $B^{\prime}_1$ we get a different code. Similarly if we select the set ${\{A^{\prime}_u\}}_{u=2}^{2a}$ in a different manner we also get a different code. So the codes described in \cite{YGT2} and \cite{YGT4} are proper subclasses of the class of CUW-SSD codes.

\section{An upper bound on the rate of UW-SSD codes}
\label{sec4}
In this section we show that for  arbitrary $2^a\times 2^a$ UW-SSD codes (not necessarily CUW-SSD codes) the rate $\frac{K}{2^a}$ in complex symbols per channel use is upper bounded by $\frac{2a}{2^a}=\frac{a}{2^{a-1}}$ which is larger than the upper bound for CODs which is $\frac{a+1}{2^a}$. Our upper bound proved in this section implies that the CUW-SSD codes constructed in previous section are rate-optimal.\\

Towards establishing an upper bound we first rewrite \eqref{ldcuwnssd} as
\begin{equation}
\label{doublenormal}
S_N=(x_{1I}I_n +\sum_{i=2}^{K} x_{iI}A_{iI})+A_{1Q}(x_{1Q}I_n+\sum_{i=2}^{K}x_{iQ} A^\prime_{iQ})
\end{equation}
where
\begin{equation*}
\label{secondnormal}
A^\prime_{iQ}= A_{1Q}^HA_{iQ}, ~~~~ 2 \leq i \leq K, \mbox{  with } A^\prime_{1Q}=I_n.
\end{equation*}
Now, if the code given by \eqref{ldcuwnssd} is UW-SSD then so is the code given by \eqref{doublenormal} and hence an upper bound on the rate of the UW-SSD codes of the form \eqref{doublenormal} is also an upper bound on the rate of the UW-SSD codes of the form \eqref{ldcuwnssd} and hence of the UW-SSD codes of the form \eqref{ldcuwssd}. Now, we proceed to obtain an upper bound on the rate of the code given by \eqref{doublenormal} when it is UW-SSD. When \eqref{doublenormal} is UW-SSD the following relations hold:
\begin{eqnarray*}
\label{a6}
A_{1I}=I_n, \quad A_{iI}^{H}=-A_{iI}\quad \textrm{for  }\, 2\leq i\leq K \\
\label{a7}
A_{iI}A_{jI}=-A_{jI}A_{iI}\quad \textrm{for  }\, 2\leq i\neq j\leq K \\
\label{a8}
{A^{\prime}}_{1Q}=I_n, \quad {A^{\prime}}_{iQ}^{H}=-{A^{\prime}}_{iQ}\quad \textrm{for  }\, 2\leq i\leq K \\
\label{a9}
{A^{\prime}}_{iQ}{A^{\prime}}_{jQ}=-{A^{\prime}}_{jQ}{A^{\prime}}_{iQ}\quad \textrm{for}\, 2\leq i\neq j\leq K \\
\label{a10}
A_{iQ}=A_{1Q}{A^{\prime}}_{iQ}\quad \textrm{for}\, 2\leq i\leq K.
\end{eqnarray*}
These relations can be proved by straight forward substitution of the weight matrices in to the set of equations given by \eqref{A2}.

The following three lemmas concerning the representations of groups will be used to prove our upper bound.
\begin{lemma}[Schur's Lemma]
\label{schurs_lemma}
For a finite group $G$, if $\{A_g \in M^{n\times n} | g \in G  \}$ is a unitary  matrix representation and $P\in M^{n\times n}$ is a nonsingular matrix that commutes with all $A_g, ~ g\in G,$ then $P= \lambda I_n$ for some non-zero $\lambda \in \mathbb{R}$.
\end{lemma}
\begin{lemma}
\label{cliff_group}
For the finite group $G_L$ of \eqref{doublecover} if there exist a matrix $P$ which commutes with all the representation matrices of the generators of $G_L$ then it commutes with all the representation matrices, i.e.,
\[
\begin{array}{l}
PA(\gamma_i)=A(\gamma_i)P, ~~ \forall i=1,2,\cdots,L, ~~ P \notin R(G_L) \\
\Longrightarrow PA(\gamma)=A(\gamma)P, ~~ \forall \gamma \in G_L.
\end{array}
\]
\end{lemma}
\begin{proof}
Let for an arbitrary element $\gamma$ of $R(G_L)$, the representation in terms of those of the generators  be
\[
A(\gamma)= A(\gamma_{i1})A(\gamma_{i2}) \cdots A(\gamma_{iL}).
\]
Then,
\[
\begin{array}{l}
PA(\gamma) = PA(\gamma_{i1})\cdots A(\gamma_{iL})=A(\gamma_{i1})P \cdots A(\gamma_{iL}) \\
=A(\gamma_{i1}) \cdots A(\gamma_{iL})P=A(\gamma)P.
\end{array}
\]
\end{proof}
\begin{lemma}
\label{unique_last}
If ${\{A_{i}\}}_{i=1}^{2a+1}$ is a $2^a\times 2^a$ irreducible representation of $CA_{2a+1}$  and for a $M \in M^{2^a\times 2^a}$,  $$\{M\}\bigcup\{{\{ A_{i} \}}_{i=1}^{2a+1}\setminus\{A_{k}\}\}$$
is also an irreducible representation of $CA_{2a+1}$, then, $ M = \pm A_{k} $.
\end{lemma}
                                                                                
\begin{proof} Since $2a+1$ is an odd number, the product of the representation matrices of the generators of the $CA_{2a+1}$ commutes with all the generators (Proposition A.2 of \cite{TiH}). Hence from Lemma~\ref{cliff_group}, this product term commutes with all the elements of the finite group generated by the generators of $CA_{2a+1}$. Then, from Schur's Lamma it follow that,\\
\begin{equation*}
\label{a13}
{\prod}_{i=1}^{2a+1}A_{i}={\lambda}_{1}I_{n\times n}\quad \textrm{for some}\, {\lambda}_{1}\, \in\, \{+1,-1\}. 
\end{equation*}
By the same argument it follows that,
\begin{equation*}
\label{a14}
{\prod}_{i=1,i\neq k}^{2a+1}A_{i}M={\lambda}_{2}I_{n \times n}\quad \textrm{for some}\, {\lambda}_{2}\, \in\, \{+1,-1\}. 
\end{equation*}
From the above two equations we have,
\begin{eqnarray*}
\left.\begin{array}{cccc}
{\prod}_{i=1}^{2a+1}A_{i} &= &\frac{\lambda_{1}}{\lambda_{2}}{\prod}_{i=1,i\neq k}^{2a+1}A_{i}M \\
\Rightarrow \, A_{k}& = & \pm \frac{\lambda_{1}}{\lambda_{2}}M \\
\Rightarrow \, A_{k}& = & \lambda M & \textrm{for some } \lambda \in \{+1,-1\}\\
\Rightarrow \, M &=& \pm \,A_{k}.
\end{array}\right. 
\end{eqnarray*}
\end{proof}
Now we are ready to prove the main result of this paper, which is an achievable upper bound on the rate of the US-SSD codes (not necessarily CUW-SSD codes) which is larger than that of the CODs. To the best our knowledge though SSD codes with rates meeting this bound have been reported no where this bound has been proved.

\begin{theorem}
\label{ubthm}
The rate $\frac{K}{2^a}$ of an $2^a\times 2^a$ UW-SSD code given by \eqref{doublenormal} is upper bounded by 
\[
\frac{K}{2^a} \leq \frac{2a}{2^a}= \frac{a}{2^{a-1}}.
\]
\end{theorem}
\begin{proof}
Since in \eqref{doublenormal} the set of matrices $A_{iI}, ~ 2 \leq i \leq K,$ constitute an Hurwitz- Radon family of matrices, for $n=2^a$, we have,
\begin{equation}
\label{b15}
K \leq 2a+2. 
\end{equation}
{\bf Claim 1  $K \neq 2a+2$:} We prove this claim by contradiction- suppose  \eqref{doublenormal} is of rate $\frac{K}{n}$ where $K=2a+2$ and $n=2^{a}$. 

Since the code is SSD, the set ${\{A_{iI}\}}_{i=2}^{2a+2}$ is a set of skew-Hermitian anticommuting unitary matrices. Hence, they represent an irreducible representation of the generators of $CA_{2a+1}$. Also, by putting $j=1$ and $2 \leq i$ in \eqref{A21}, we get,
\[ A_{iI}^{H}A_{1Q}+A_{1Q}^{H}A_{iI}=0, \quad 2\leq i \leq 2a+2 \]
\begin{equation}
\label{b17}
\left.
\begin{array}{c}
\Rightarrow \quad A_{iI}A_{1Q}=A_{1Q}^{H}A_{iI}\\
\Rightarrow \quad A_{1Q}A_{iI}=A_{iI}A_{1Q}^{H} \end{array} \right\}.
\end{equation}
Now from the set ${\{A_{iI}\}}_{i=2}^{2a+2}$ we construct another set ${\{B_{i}\}}_{i=2}^{2a+2}$, given by
\begin{equation}
\label{b18}
\left.\begin{array}{c}
 B_{i}=A_{2I}A_{iI} \quad  3\leq i \leq 2a+2 \\
\textrm{and} \quad B_{2}=j{\prod}_{i=3}^{2a+2}A_{iI} \end{array} \right\}.
\end{equation}

Now it can easily be verified that this new set ${\{B_{i}\}}_{i=2}^{2a+2}$ is also a set of skew-Hermitian anticommuting unitary matrices. Hence this also represents an irreducible representation of the generators of $CA_{2a+1}$. Further, for $2\leq i \leq 2a+2 $, using \eqref{b17} and \eqref{b18}, we have
\begin{eqnarray*}
\label{b19}
A_{1Q}B_{i}=A_{1Q}A_{2I}A_{iI}=A_{2I}A_{1Q}^{H}A_{iI} \\
=A_{2I}A_{iI}A_{1Q}=B_{i}A_{1Q}.
\end{eqnarray*}
                                                                                
Moreover,
\begin{equation}
\label{b20}
A_{1Q}B_{2}=A_{1Q}j{\prod}_{i=2}^{2a+2}A_{iI}=j{\prod}_{i=2}^{2a+2}A_{iI}A_{1Q}=B_{2}A_{1Q}
\end{equation}
which follows from repeated use of \eqref{b17} noting that there are even no of terms in the product. So $A_{1Q}$ is a matrix that commutes with all the representation matrices of the generators of the $CA_{2a+1}$. Hence using Lemma~\ref{schurs_lemma} and Lemma~\ref{cliff_group} we have,
\[A_{1Q}=\lambda I_{n\times n}.\]
If $\lambda \in \mathbb{R}$, then $A_{1I}=\lambda A_{1Q}$ which contradicts \eqref{nopathology}. On the other hand, if $\lambda \notin \mathbb{R}$ but $\lambda \in \mathbb{C}$ then condition \eqref{b17} is violated. This means there does not exists an $A_{1Q}$ which satisfies all the conditions and hence a code of the assumed rate does not exist. So $K\neq 2a+2$.\\

{\bf Claim 2 $K \neq 2a+1$:} The proof for this claim is given in Appendix-III.\\

From these two claims and \eqref{b15} it follows that
\begin{equation*}
\label{upperbound}
K \leq 2a \mbox{ and hence } \frac{K}{2^a} \leq \frac{2a}{2^a}= \frac{a}{2^{a-1}}.
\end{equation*}
\end{proof}
\section{Diversity and Coding gain of CUW-SSD codes}
\label{sec5}
We have seen in Theorem \ref{gainnochange} that the coding gain of a UW-SSD does not change when normalized. Hence, for a CUW-SSD code $S$ the expression given by \eqref{dpnuwssd} can be used. Towards this end, since CUW-SSD codes satisfy the sufficient conditions  \eqref{suffcondssd}, we have  
\begin{eqnarray}
\label{2aq1}
\begin{array}{l}
 A_{iI}^{H}A_{iQ}+A_{iQ}^{H}A_{iI} \\
 = A_{iI}^{H}A_{1Q}A_{iI}+A_{iI}^{H}A_{1Q}^{H}A_{iI}  \\
 = A_{iI}^{H}A_{1Q}A_{iI}+A_{iI}^{H}A_{1Q}A_{iI}    \\
 = A_{iI}^{H}A_{iI}A_{1Q}+A_{iI}^{H}{H}A_{iI}A_{1Q}  \\
 = \left(A_{iI}^{H}A_{iI}+A_{iI}^{H}A_{iI}\right)A_{1Q} \\
 = 2I_nA_{1Q} = 2A_{1Q}~~ \forall \quad k \geq i\neq j\geq 1.
\end{array}
\end{eqnarray}

Using \eqref{2aq1} in \eqref{dpnuwssd} we get,

{\footnotesize 
\begin{equation} 
\label{div_prod1}
\mathbf{DP}(S)=\frac{1}{2\sqrt{n}} \min_{\triangle\mathbf{x}\neq \mathbf{0}}  \Bigg\lvert  \mathrm{det} \sum_{i=1}^{k} \Bigg[  {|\triangle x_{i}|}^{2}I_{n} +2\triangle x_{iI}\triangle x_{iQ}A_{1Q} \Bigg]  \Bigg \rvert
\end{equation}
}
         
The above expression shows that the Hermitian matrix $A_{1Q}$ is special among all the weight matrices of the code, in the sense that this alone influences the coding gain. For this reason we give the name {\it the discriminant of $S$} to it. Since the discriminant is unitary, it is diagonalizable,  say, $A_{1Q}=\mathbf{E}\Lambda\mathbf{E^{-1}},$ where $\mathbf{E}$ is the matrix containing the eigenvectors of $A_{1Q}$ and $\Lambda $ is the diagonal matrix containing the eigenvalues of $A_{1Q}$. Now as eigenvalues of unitary matrix lie on the unit circle and eigenvalues of Hermitian matrix are all real, the entries of $\Lambda $ are $\pm 1$ only. Using this information in \eqref{div_prod1} we have

{\footnotesize
\begin{displaymath}
\mathbf{DP}(S)=\frac{1}{2\sqrt{n}}  \min_{\triangle\mathbf{x}\neq \mathbf{0}}  \Bigg\lvert  \mathrm{det} \sum_{i=1}^{k} \Bigg[  {|\triangle x_{i}|}^{2}\mathbf{E}\mathbf{E^{-1}} +2\triangle x_{iI}\triangle x_{iQ}\mathbf{E}\Lambda\mathbf{E^{-1}} \Bigg]  \Bigg \rvert  
\end{displaymath}
}
\begin{displaymath}
=\frac{1}{2\sqrt{n}} \min_{\triangle\mathbf{x}\neq \mathbf{0}}  \Bigg\lvert  \mathrm{det} \sum_{i=1}^{k} \Bigg[  {|\triangle x_{i}|}^{2} I_{n} +2\triangle x_{iI}\triangle x_{iQ}\Lambda \Bigg]  \Bigg \rvert  
\end{displaymath}

\begin{displaymath}
=\frac{1}{2\sqrt{n}}  \min_{\triangle\mathbf{x}\neq \mathbf{0}}  \Bigg\lvert  \mathrm{det}\left[ \begin{array}{ccc}
\lambda_i & \ldots & 0 \\
 \vdots & \ddots & \vdots \\
 0 & \ldots & \lambda_i \end{array} \right]  \Bigg \rvert \\
\end{displaymath}
\hspace{1.5cm} where $\lambda_i  = \sum_{i=1}^{k} {\left(\triangle x_{iI}\pm\triangle x_{iQ}\right)}^{2}$.
\begin{displaymath}
=\frac{1}{2\sqrt{n}}  \min_{\triangle\mathbf{x} \,\neq\, \mathbf{0}}  \Bigg\lvert  \prod_{j=1}^{n}  \Bigg[\sum_{i=1}^{k} {\left(\triangle x_{iI}+{\left(-1\right)}^{s_{j}}\triangle x_{iQ}\right)}^{2} \Bigg]  \Bigg \rvert  
\end{displaymath}
where, $s_{i} \in \{ 0,1 \} $ depending on the eigenvalues of $A_{1Q}$. Now every term in the inner summation is $\geqslant 0$. Hence the minimum of $\mathbf{DP}(S)$ is attained when all $\triangle x_{i}$ except one is zero, leading to

{\footnotesize
\begin{equation}
\label{codgain}
\mathbf{DP}(S)=\frac{1}{2\sqrt{n}}  \min_{\triangle x_{i} \, \neq \, 0} \Bigg \lvert {\left(\triangle x_{iI}+\triangle x_{iQ}\right)}^{2m}{\left(\triangle x_{iI}-\triangle x_{iQ}\right)}^{2n-2m} \Bigg \rvert 
\end{equation} 
}
\noindent
where $A_{1Q}$ has $m$ number of $+1$s and the remaining $n-m$ number of $-1$s as eigenvalues. As can easily be seen, this code is not full diversity in general, for if $\triangle x_{i} \, \neq \, 0 $ but, $\triangle x_{iI}= \pm \triangle x_{iQ}$,  then $\mathbf{DP}(S)=0$. This proves the following theorem giving a  set of necessary and sufficient conditions for a code CUW-SSD code $S$ to have full-diversity.
\begin{theorem}
\label{fullrankcond}
Let $S$ given in \eqref{ldcuwnssd} be a CUW-SSD code with the variables $x_i, ~~ i=1,2,\cdots,K$ taking values from a complex signal set $\mathbb S$. Also, let
\begin{equation*}
\label{diffsigset}
\Delta {\mathbb S} = \left\{ a-b | a,b \in {\mathbb S} \right\}  
\end{equation*}
be the difference signal set of $\mathbb S$. Then, $S$ will have full-diversity if and only if the difference signal set $\Delta {\mathbb S}$ does not have any point on the lines that are at $\pm 45 $ degrees in the complex plane apart from the origin. 
\end{theorem}
From the expression for the diversity product \eqref{codgain} we see that the coding gain depends not only on the signal set from which the variables take values, it depends also on the discriminant $A_{1Q}$ of the code $S$ via $m$. So, the problem of maximizing the coding gain involves the proper choice of the discriminant for the code as well as the signal set. If the discriminant is chosen such that it is  traceless (i.e., it has trace equal to zero or equivalently it has the same number of +1 and -1's as eigenvalues), then $m=n/2$ and \eqref{codgain} reduces to 

\begin{equation}
\label{codgaintraceless}
\mathbf{DP}(S)=\frac{1}{2\sqrt{n}}  \min_{\triangle x_{i}  \neq  0} \Bigg \lvert {\left(\triangle x_{iI}^2   -\triangle x_{iQ}^2\right)}^n  \Bigg \rvert
\end{equation}
which does not depend on the discriminant of the code. 

Notice that with the traceless condition, we have the discriminant to be a traceless, unitary, anti-Hermitian matrix commuting with all $A_{iI},~~i=1,2,\cdots,K$. We conjecture the following:

\noindent
{\bf Conjecture:} For a given signal set the diversity product expression \eqref{codgain} is maximum when $2m=n$, i.e., when the discriminant of the code is traceless. 

If this conjecture is true, we are left with only the problem of finding the signal set $\mathbf{S}$ such that the $\mathbf{DP}$ is maximized.
\subsection{Diversity product calculations}
In this subsection we will show that it is possible to achieve the same diversity product as that of MDC-QOD codes described in \cite{WWX1} through our code, for both rectangular and square-derived QAM constellations. Towards this end, let us first consider the rectangular case. Say, $y_i=y_{iI}+jy_{iQ}\in \mathcal{A}_1, \forall i$. Let us form the complex symbols, $x_i=x_{iI}+jx_{iQ}, 1\leq i \leq K$ in the following way,
\begin{equation*}\left[\begin{array}{c}
x_{iI}\\
x_{iQ}
\end{array}\right]
=T^{-1}\left[\begin{array}{c}
y_{iI}\\
y_{iQ}
\end{array}\right] \forall i
\end{equation*}
where,
\begin{equation*}
T=\left[\begin{array}{cc}
\frac{1}{\sqrt{2}} & \frac{1}{\sqrt{2}} \\
\frac{1}{\sqrt{2}} & -\frac{1}{\sqrt{2}}
\end{array}\right]
\end{equation*}
and construct CUW-SSD code with these variables. Then using \eqref{codgaintraceless} it can be shown that the diversity product of our code is dependent on the CPD (Co-ordinate Product Distance) of $\mathcal{A}_1$, i.e.,
\begin{equation*}
DP=\frac{1}{2 \sqrt{n}} \min_{\triangle y_i\neq 0}{\lvert 2\triangle y_{iI}\triangle y_{iQ}\rvert}^{\frac{1}{2}}.
\end{equation*}
It can be further shown that this is exactly equal to the diversity product of a \textit{normalized} (the codeword is multiplied by an appropriate constant so that the total transmitted power is $N_t^2$, where $N_t$ is the number of transmit antennas) CIOD code whose variables takes their values from $\mathcal{A}_1$. Note that $T$ is an unitary matrix. Hence the total transmitted power per codeword is same. From the above discussion we can see that if we are going to use a rectangular QAM constellation say, $\mathcal{A}_0$, then we need to find a linear transformation matrix  $U$, such that the transformed constellation $\mathcal{A}_1$ in Fig \ref{figct} have maximum CPD. Now if we form the constellation $\mathcal{A}_2$ as shown in Fig \ref{figct} and allow our code variables to take value from this constellation then our code will be achieving the same diversity product as a normalized CIOD can achieve using $\mathcal{A}_1$. 
Theorem $6$ in \cite{WWX1} gives the linear transformation matrix $U$ we need. We illustrate the method of obtaining $U$ when one uses the rectangular constellation

\begin{eqnarray*}
\mathcal{A}_0 & = & \Big\{(\frac{n_1 d}{2}+j\frac{n_2 d}{2}):n_i\in N_i~ \textrm{for}~ i=1,2\Big\}~~  \textrm{where} \\
N_i & \triangleq & \Big\{-(2N_i^{\prime}-1),-(2N_i^{prime}-3) \cdots \\
 & &    -1,1, \cdots (2N_i^{\prime}-3),(2N_i^{\prime}-1)\Big\}
\end{eqnarray*}
where $N_i^{\prime}$ are positive integers and $d$ is a real positive constant that is used to adjust total energy. Now if $\varepsilon_1=\frac{2N_1^2-1}{2(2N_1^2+2N_2^2-1)}$, $\varepsilon_2=\frac{2N_2^2-1}{2(2N_1^2+2N_2^2-1)}$, $\alpha=\tan^{-1}\left(\frac{1}{\sqrt{\varepsilon_1 \varepsilon_2}}\right)$, $\theta_1=\tan^{-1}\left(\frac{\sqrt{5}-1}{2}\sqrt{\frac{\varepsilon_1}{\varepsilon_2}}\right)$ and $\theta_2=(\alpha-\theta_1)$ then $U$ is given by,

\begin{equation*}
U=\left(\begin{array}{cc}
\frac{\cos(\theta_1)}{\sqrt{2\varepsilon_1}} & \frac{\sin(\theta_1)}{\sqrt{2\varepsilon_2}}\\
\frac{-\sin(\theta_2)}{\sqrt{2\varepsilon_1}} & \frac{\cos(\theta_2)}{\sqrt{2\varepsilon_2}}
\end{array}\right).
\end{equation*}
\begin{figure}[hbt]
\centering
\includegraphics[width=6.0cm,height=6.0cm]{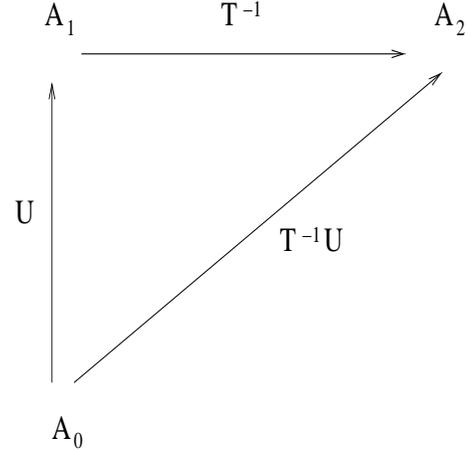}
\caption{Linear transformations of constellations }
\label{figct}
\end{figure}

Following the above mentioned method we have calculated diversity products of our code for $8(4\times 2)$ and $32(8\times 4)$ rectangular constellations. For square or square-derived constellations we follow the same procedure as explained above. The only difference is now we use the linear transformation matrix $U$ given above with $N_1=N_2$, where $N_1^2$ is the nearest even square that is greater than or equal to the size of the constellation. In Table 1 below we compare the diversity product of our code (CUW-SSD) with that of MDC-QOD for various constellations. All these calculations were done assuming total constellation energy equal to 1. \\

TABLE I : Diversity Product comparison.\\
\begin{tabular}[c]{|c|c|c|c|}
\hline
Constellation: & $4$-QAM &  $8$-QAM & $32$-QAM  \\
Square derived  &  &  &    \\
\hline
 MDC-QOD & $.1672$ & $.0757$ & $.0187$  \\
\hline
CUW-SSD & $.1672$ & $.0757$ & $.0187$  \\
\hline \hline
Constellation: & $4$-QAM &  $8$-QAM & $32$-QAM  \\
Rectangular QAM  &  &  &    \\
\hline
 MDC-QOD & $.1672$ & $.0699$ & $.0167$  \\
\hline
CUW-SSD & $.1672$ & $.0699$ & $.0167$  \\
\hline
\end{tabular}
\vspace{1cm}

We see that diversity product of our codes matches exactly with those of comparable MDC-QOD codes. Hence it is expected that the error performance will also be same. This has been verified through simulation results given in the following subsection.

\subsection{Simulation Results for our SSD codes}
\label{subsec6}
In this subsection we provide some simulation results. The simulations have been carried out for one receive antenna only. We have compared the error performance of our code with the best known SSD code in the literature\cite{WWX1}. We performed simulations for 2,3 and 5 bits per channel uses respectively. For 3-bits per channel use and 5-bits per channel uses we have used both rectangular and square derived QAM constellations. We derive a "square derived $q$-QAM" in the following way: We take a nearest square number $p$ which is greater than $q$, and from the $p$-QAM delete the larger energy $p-q$ points and then translate the resulting constellation so that its CG is at the origin. In Fig. \ref{Sq-derived8qam}  and Fig. \ref{Sq-derived32qam} we have shown square derived $8$-QAM and square derived $32$-QAM constellations respectively. In Fig. \ref{ssd_2bpcu} we compared the performance of our code with MDC-QOD at 2 bits per channel use (We used $4$-QAM) and it matches with the theoretical results suggested by the fact that the diversity product is same for both the codes as shown in TABLE I. Now for spectral efficiencies of $3$-bits per channel use and $5$-bits per channel use we see from the Table I that both for rectangular QAM and square derived QAMs the diversity product of our code is same to that of MDC-QOD codes \cite{WWX1}. Hence we expect that the error performance of both the codes should be same. We see in Fig. \ref{ssd_3bpcu} and Fig. \ref{ssd_5bpcu} respectively that this is indeed the case.


\begin{figure}[hbt]
\centering
\includegraphics[width=6.0cm,height=6.0cm]{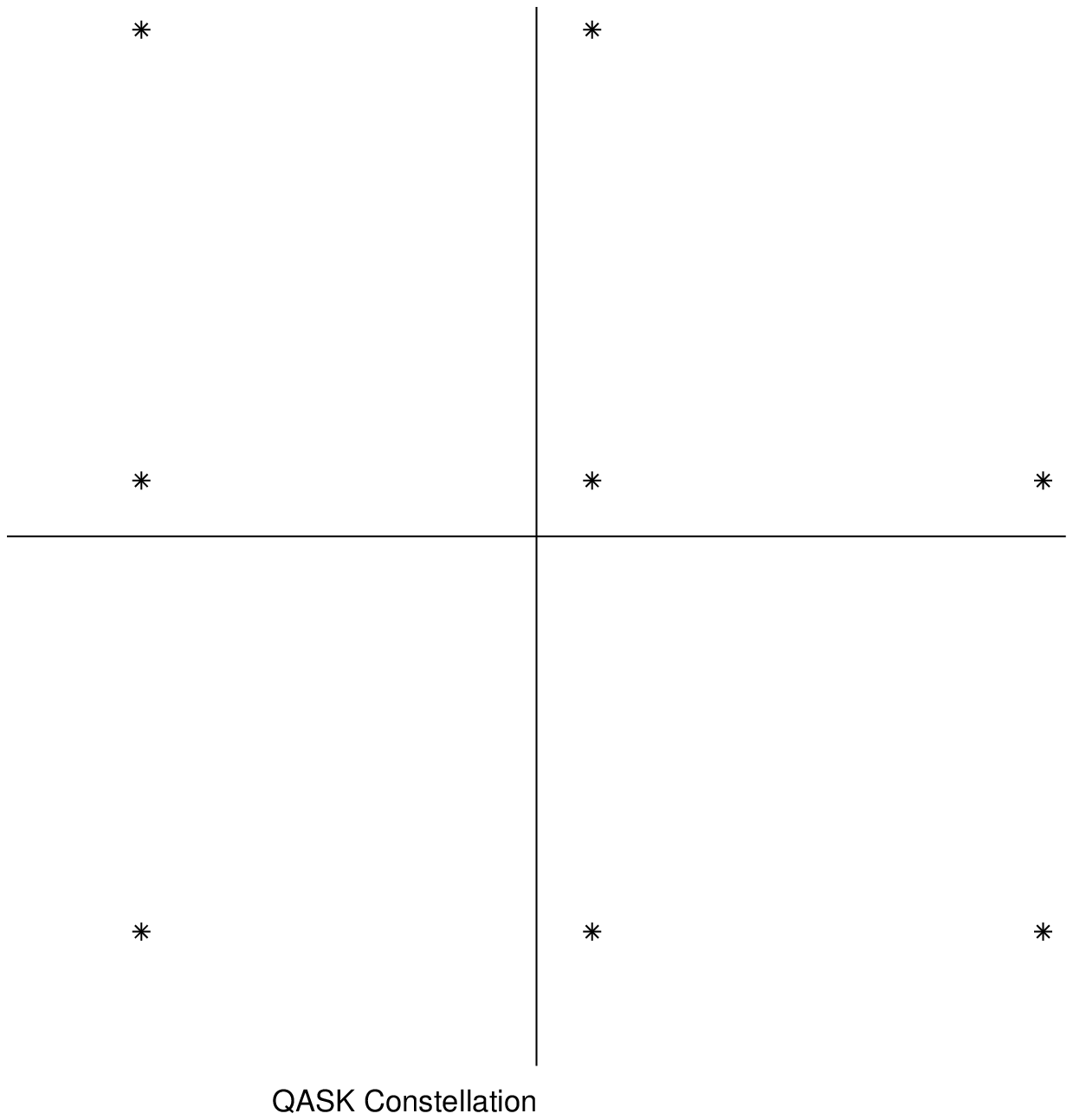}
\caption{Square derived 8-QAM constellation. }
\label{Sq-derived8qam}
\end{figure}

\begin{figure}[hbt]
\centering
\includegraphics[width=7.0cm,height=7.0cm]{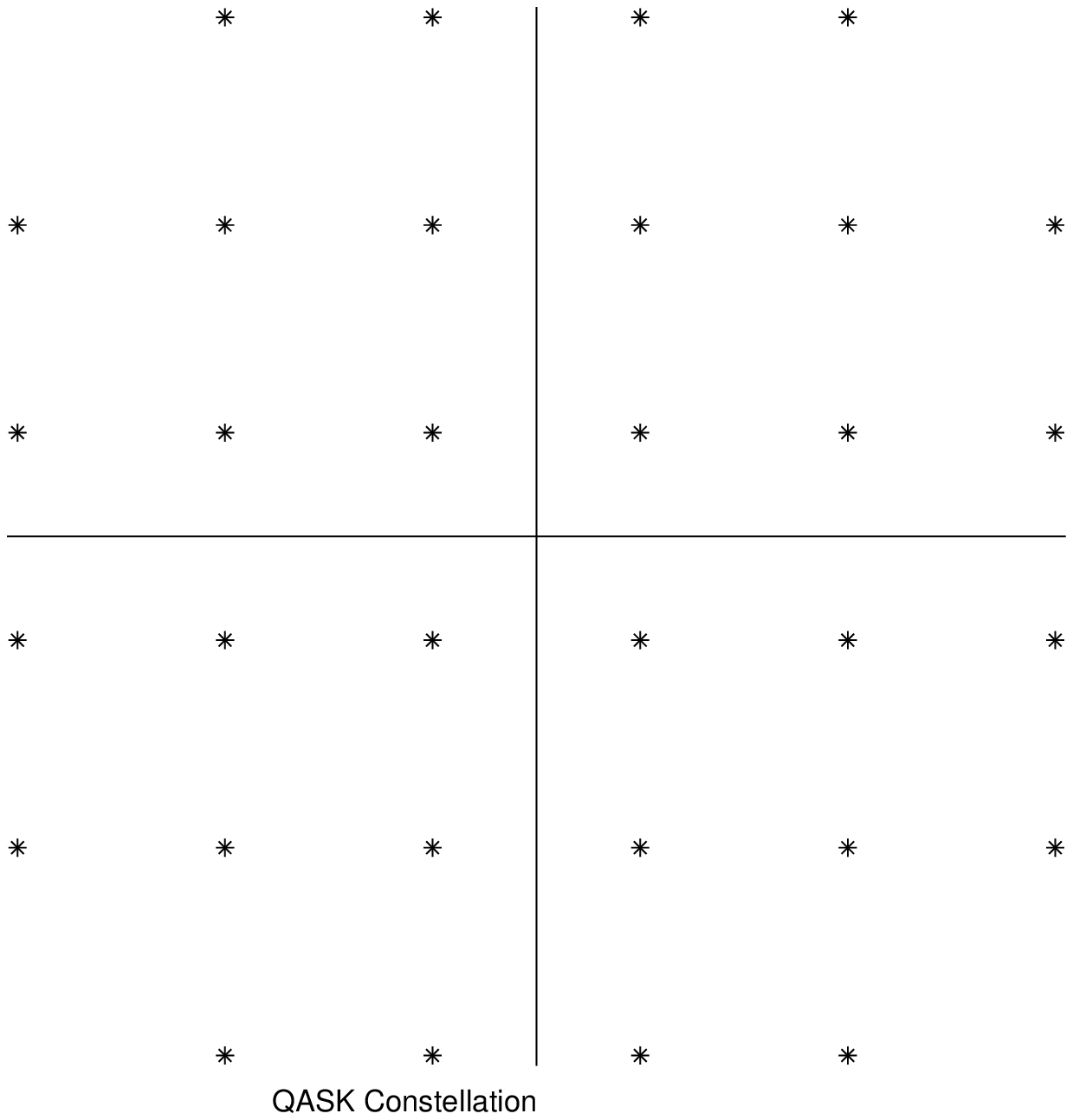}
\caption{Square derived 8-QAM constellation. }
\label{Sq-derived32qam}
\end{figure}

\begin{figure}[hbt]
\centering
\includegraphics[width=8.0cm,height=8.0cm]{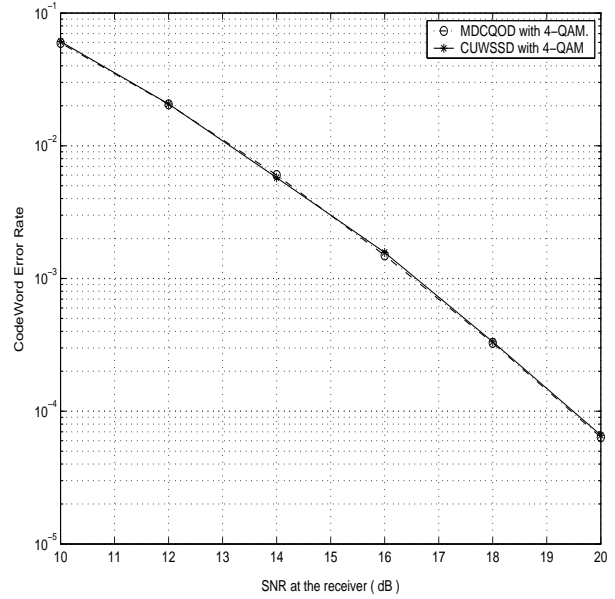}
\caption{Comparison of CUW-SSD code's performance with MDC-QOD code at 2 Bits per Channel use. }
\label{ssd_2bpcu}
\end{figure}

\begin{figure}[[h]
\centering
\includegraphics[width=7.0cm,height=7.0cm]{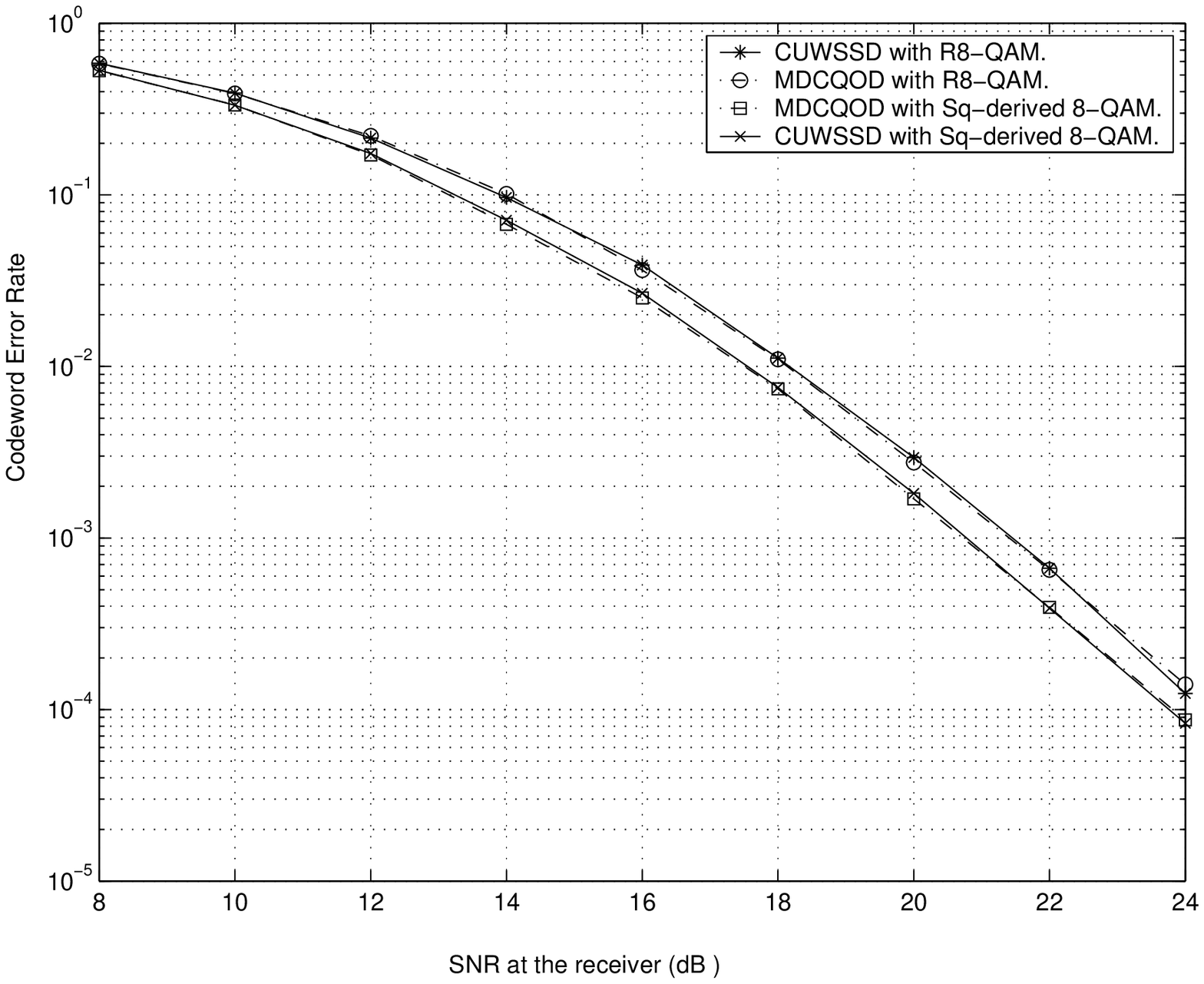}
\caption{Comparison of CUW-SSD code's performance with MDC-QOD code at 3 Bits per Channel use with Square-derived and Rectangular 8-QAM constellation. }
\label{ssd_3bpcu}
\end{figure}

\begin{figure}[hbt]
\centering
\includegraphics[width=8.0cm,height=8.0cm]{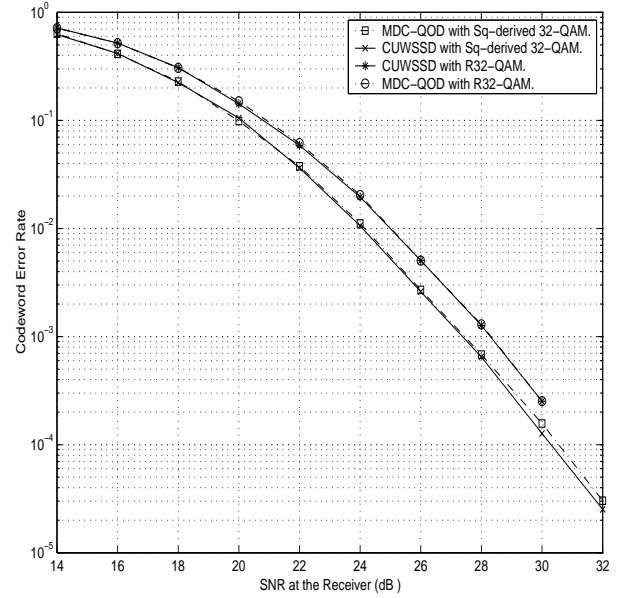}
\caption{Comparison of CUW-SSD code's performance with MDC-QOD code at 5 Bits per Channel use with Square-derived and Rectangular 8-QAM constellation. }
\label{ssd_5bpcu}
\end{figure}
                                                                                
\section{Non-Unitary Weight- SSD codes from Clifford algebras}
\label{sec6}
In this section we obtain a class of non-unitary weight SSD codes from CUW-SSD codes by employing linear transformations on the weight matrices. 

\begin{definition}
For a normalized UW-SSD code 
\begin{equation*} 
\label{uwssd} 
S=\sum_{i=1}^{K}\left(x_{iI}A_{iI}+x_{iQ}A_{iQ}\right)
\end{equation*}
and a pair of non-zero real numbers $\alpha, \beta$, define the {\it Transformed Non-Unitary} code to be 

\begin{equation} 
\label{uwssdT} 
S_T=\sum_{i=1}^{k}\left(x_{iI}T_{iI}+x_{iQ}T_{iQ}\right)
\end{equation}
where  
\begin{eqnarray}
\label{transform}
\begin{array}{rl}
T_{iI} & = \alpha A_{iI}+\beta A_{iQ} \\
T_{iQ} & = \alpha A_{iI}-\beta A_{iQ}.
\end{array}
\end{eqnarray} 
\end{definition}

From
\begin{eqnarray*}
(\alpha I_n +\beta A_{1Q})^H(\alpha I_n + \beta A_{1Q}) \\
=(\alpha I_n +\beta A_{1Q})(\alpha I_n + \beta A_{1Q}) \\
=\alpha^2I_n +\beta^2 A_{1Q}^2+2\alpha \beta A_{1Q} \\
=\alpha^2I_n +\beta^2 I_n+2\alpha \beta A_{1Q} \\
=(\alpha^2 +\beta^2) I_n+2\alpha \beta A_{1Q}
\end{eqnarray*}
it follows that $\alpha A_{1I}+\beta A_{1Q}$ is not unitary unless $\alpha= \beta =0$ which ensures that $S_T$ is not a UW-code. 
\begin{theorem}
\label{nussdthm}
The Transformed Non-unitary code given by \eqref{uwssdT} is SSD.
\end{theorem}
\begin{proof}
Observe that $T_{1I}$ and $T_{1Q}$ are Hermitian and $T_{iI},~ T_{iQ}, ~~ i=2,3,\cdots,K$ are anti-Hermitian. By construction $A_{1Q}$ and $A_{1I}=I$ commute with all $A_{iI}, A_{iQ}, ~~i=2,3,\cdots,K$ and hence $T_{1I}$ and $T_{1Q}$ commute with all $T_{iI}, ~T_{iQ}, ~~ i=2,3,\cdots,K$. Now, given
\begin{eqnarray*}
\label{givenA2}
\left. \begin{array}{rl}
A_{iI}^{H}A_{jQ} + A_{jQ}^{H}A_{iI} & =0 \\
A_{iI}^{H}A_{jI} + A_{jI}^{H}A_{iI} & =0 \\
A_{iQ}^{H}A_{jQ} + A_{jQ}^{H}A_{iQ} & =0
\end{array} \right\}  ~~~  1 \leq i \neq j \leq K,
\end{eqnarray*}
we need to prove that
\begin{eqnarray}
\label{toproveA2}
\left. \begin{array}{rl}
T_{iI}^{H}T_{jQ} + T_{jQ}^{H}T_{iI} & =0 \\
T_{iI}^{H}T_{jI} + T_{jI}^{H}T_{iI} & =0 \\
T_{iQ}^{H}T_{jQ} + T_{jQ}^{H}T_{iQ} & =0
\end{array} \right\}  ~~~  1 \leq i \neq j \leq K.
\end{eqnarray}

We prove below only the second equation of \eqref{toproveA2} and the proof for the remaining two equations are similar.\\
{\bf Case(i) $i=1$ or $j=1$:} Let $i=1$ and $1 \leq j \leq K$. Then \eqref{casei} at the top of the next page  shows that $T_{1I}^{H}T_{jI} + T_{jI}^{H}T_{1I} =0$.
\begin{figure*}
\begin{equation}
\label{casei}
\begin{array}{l}
T_{1I}^{H}T_{jI} + T_{jI}^{H}T_{1I} \\
 = (\alpha A_{1I}+\beta A_{1Q})^H(\alpha A_{jI}+\beta A_{jQ})+(\alpha A_{jI}+\beta A_{jQ})^H(\alpha A_{1I}+\beta A_{1Q}) \\ 
  =  (\alpha A_{1I}+\beta A_{1Q})(\alpha A_{jI}+\beta A_{jQ})-(\alpha A_{jI}+\beta A_{jQ})(\alpha A_{1I}+\beta A_{1Q}) \\
  = 0 ~~~ \mbox{ since } (\alpha A_{1I}+\beta A_{1Q}) \mbox { commutes with all } (\alpha A_{jI}-\beta A_{jQ}).
\end{array} 
\end{equation} \hrule
\end{figure*}
\noindent
{\bf Case(ii) $1 \leq i \neq j \leq K$:} For this case \eqref{caseii} at the top of the next page shows that $T_{iI}^{H}T_{jI} + T_{jI}^{H}T_{iI}=0$.
\begin{figure*}
\begin{equation}
\label{caseii}
\begin{array}{l}
T_{iI}^{H}T_{jI} + T_{jI}^{H}T_{iI} \\ 
= (\alpha A_{iI}+\beta A_{iQ})^H(\alpha A_{jI}+\beta A_{jQ})+(\alpha A_{jI}+\beta A_{jQ})^H(\alpha A_{iI}+\beta A_{iQ}) \\
=-\left[ (\alpha A_{iI}+\beta A_{iQ})(\alpha A_{jI}+\beta A_{jQ})+(\alpha A_{jI}+\beta A_{jQ})(\alpha A_{iI}+\beta A_{iQ})\right] \\
=-\left[ (\alpha A_{iI}+\beta A_{iQ})\alpha A_{jI}+ (\alpha A_{iI}+\beta A_{iQ}) \beta A_{jQ}+(\alpha A_{jI}+\beta A_{jQ})(\alpha A_{iI}+\beta A_{iQ})\right] \\

=-\left[-\alpha A_{jI} (\alpha A_{iI}+\beta A_{iQ})- \beta A_{jQ} (\alpha A_{iI}+\beta A_{iQ})+(\alpha A_{jI}+\beta A_{jQ})(\alpha A_{iI}+\beta A_{iQ})\right] \\

=-\left[-(\alpha A_{jI}+ \beta A_{jQ} ) (\alpha A_{iI}+\beta A_{iQ})+(\alpha A_{jI}+\beta A_{jQ})(\alpha A_{iI}+\beta A_{iQ})\right] \\
=0.
\end{array}
\end{equation} \hrule
\end{figure*}
\end{proof}
The following theorem obtains a necessary and sufficient condition for a transformed NU-SSD code to be a PSSD code.
\begin{theorem}
\label{nupssdthm}
The Transformed Non-Unitary code given by \eqref{uwssdT} is a PSSD code iff $\alpha \neq \pm \beta$ in \eqref{transform}. Equivalently, the Transformed Non-Unitary code of \eqref{uwssdT} is NU-COD iff $\alpha = \pm \beta$ in \eqref{transform}.
\end{theorem}
\begin{proof}
We need to show that $\alpha = \pm \beta$ iff
\begin{equation*}
T_{iI}^HT_{iQ}+T_{iQ}^HT_{iI}=0, ~~~~ i=1,2,\cdots,K.
\end{equation*}
{\bf Case (i)  $i=1$:} In this case,
{\small
\begin{eqnarray*}
\begin{array}{l}
T_{1I}^HT_{1Q}+T_{1Q}^HT_{1I} \\
=(\alpha I+\beta A_{1Q})^H(\alpha I-\beta A_{1Q})+(\alpha I-\beta A_{1Q})^H(\alpha I+\beta A_{1Q}) \\
=(\alpha I+\beta A_{1Q})(\alpha I-\beta A_{1Q})+(\alpha I-\beta A_{1Q})(\alpha I+\beta A_{1Q}) \\
=\alpha^2 I -\beta^2 A_{1Q}^2+\alpha^2 I -\beta^2 A_{1Q}^2 \\
=2(\alpha^2- \beta^2)I
\end{array}
\end{eqnarray*} 
}
which is zero iff $\alpha = \pm \beta$.\\
{\bf Case (ii)  $i \geq 2$:} In this case, $T_{iI}^HT_{iQ}+T_{iQ}^HT_{iI}$
\begin{eqnarray*}
\begin{array}{ll}
= & (\alpha A_{iI}+\beta A_{iQ})^H(\alpha A_{iI}-\beta A_{iQ})+  \\ 
  & (\alpha A_{iI}-\beta A_{iQ})^H(\alpha A_{iI}+\beta A_{iQ}) \\
= & -\left[ (\alpha A_{iI}+\beta A_{iQ})(\alpha A_{iI}-\beta A_{iQ})+\right. \\
  & \left. (\alpha A_{iI}-\beta A_{iQ})(\alpha A_{iI}+\beta A_{iQ}) \right] \\
= & -\left[\alpha^2 A_{iI}^2+\alpha \beta A_{iQ}A_{iI} -\alpha \beta A_{iI}A_{iQ}-\beta^2 A_{iQ}^2 +\right.\\ 
  & \left. \alpha^2 A_{iI}^2-\alpha \beta A_{iQ}A_{iI}+ \alpha \beta A_{iI}A_{iQ}+\beta^2 A_{iQ}^2   \right]\\
= & -2 \left[\alpha^2 A_{iI}^2 - \beta^2 A_{iQ}^2  \right] \\
= & 2(\alpha^2 - \beta^2)I
\end{array}
\end{eqnarray*}

which is zero iff $\alpha = \pm \beta$.
\end{proof}

\begin{definition}
The Non-Unitary SSD codes obtained from the $2^a$-CUW-SSD codes under the transform given by \eqref{transform} are called (i) Clifford NU-CODs and abbreviated as $a$-CNU-CODs, if $\alpha = \pm \beta$ and (ii) Clifford Proper SSD codes, abbreviated as CP-SSD codes, if $\alpha \neq \pm \beta$.
\end{definition}
%
%
\begin{example}
\label{4by4}
Consider the $4$-CUW-SSD code defined by the following weight matrices
\begin{eqnarray*}
\begin{array}{ccc}
A_{1I}=I_{2}\otimes I_{2}, & A_{2I}=I_{2}\otimes j{\sigma}_{3}, & A_{3I}=I_{2}\otimes{\sigma}_{1}, \\
 A_{4I}=I_{2}\otimes{\sigma}_{2} &  A_{1Q}={\sigma}_{3}\otimes I_{2},  & A_{2Q}=A_{1Q}A_{2I} \\
 A_{3Q}=A_{1Q}A_{3I}, & A_{4Q}=A_{1Q}A_{4I} & 
\end{array} 
\end{eqnarray*}
to obtain the code $\mathbf{S}\left(x_{1},x_{2},x_{3},x_{4}\right)$ given by \eqref{ciodfirstcut} given at the top of the next  page.
\begin{figure*}
{\small
\begin{equation}
\label{ciodfirstcut}
\left[\begin{array}{cccc}
x_{1I}+x_{1Q}+ j(x_{2I}+x_{2Q}) & x_{3I}+x_{3Q}+j(x_{4I}+x_{4Q}) & 0 & 0 \\
-x_{3I}-x_{3Q}+  j(x_{4I}+x_{4Q}) &  x_{1I}+x_{1Q}- j(x_{2I}+x_{2Q}) & 0 & 0 \\
0 & 0 & x_{1I}-x_{1Q}+ j(x_{2I}-x_{2Q})&  x_{3I}-x_{3Q}+ j(x_{4I}-x_{4Q}) \\
0 & 0 &  -x_{3I}+x_{3Q}+ j(x_{4I}-x_{4Q})  & x_{1I}-x_{1Q}-  j(x_{2I}-x_{2Q}) 
\end{array} \right]
\end{equation} \hrule
} 
\end{figure*}
Using the transform,
\begin{equation*} \left[\begin{array}{c}
x^{\prime}_{iI}\\x^{\prime}_{iQ} \end{array} \right]  =  \left[\begin{array}{cc} 1 & 1\\ 1 & -1 \end{array} \right] \left[ \begin{array}{c} x_{iI}\\x_{iQ} \end{array} \right] \end{equation*}
\noindent
we get the corresponding $2$-CNU-SSD code given by \eqref{ciod_1} shown at the top of the next page.
\begin{figure*}
\begin{equation} 
\label{ciod_1}
\mathbf{S}\left(x^{\prime}_{1},x^{\prime}_{2},x^{\prime}_{3},x^{\prime}_{4}\right)= \left[\begin{array}{cccc}
x^{\prime}_{1I}+jx^{\prime}_{2I} & x^{\prime}_{3I}+jx^{\prime}_{4I} & 0 & 0 \\
-x^{\prime}_{3I}+jx^{\prime}_{4I} & x^{\prime}_{1I}-j x^{\prime}_{2I} & 0 & 0 \\
0 & 0 & x^{\prime}_{1Q}+jx^{\prime}_{2Q}  & x^{\prime}_{3Q}+jx^{\prime}_{4Q} \\
0 & 0 & -x^{\prime}_{3I}+x^{\prime}_{4Q}  & x^{\prime}_{1Q}-x^{\prime}_{2Q}
\end{array}\right] 
\end{equation} \hrule
\end{figure*}
Note that the code in \eqref{ciod_1} is of the form,
\[\left[\begin{array}{cccc}
\widetilde{x}_{1} & \widetilde{x}_{2} & 0 & 0\\
-{\widetilde{x}}^{\ast}_{2} & {\widetilde{x}}^{\ast}_{1} & 0 & 0 \\
0 & 0 & \widetilde{x}_{3} & \widetilde{x}_{4} \\
0 & 0 & \widetilde{x}^{\ast}_{4} & \widetilde{x}^{\ast}_{3}
\end{array}\right] \]
\noindent
where $\widetilde{x}_{i}=x^{\prime}_{iI}+jx^{\prime}_{\left(i+1\right)I},\, \textrm{ for } \,i=1,2,$   and  $\widetilde{x}_{i}=x^{\prime}_{\left(i-2\right)Q}+jx^{\prime}_{\left(i-1\right)Q}, \textrm{ for } i=2,3$. 
\end{example}                                                                   

\begin{example}
\label{4by4cpssd}
Consider the UW-SSD code defined by the following weight matrices:
                                                                                
\begin{eqnarray*}
\begin{array}{ccc}
A_{1I}=I_{2}\otimes I_{2}, & A_{2I}=I_{2}\otimes j{\sigma}_{3}, & A_{3I}=I_{2}\otimes{\sigma}_{1}, \\
 A_{4I}=I_{2}\otimes{\sigma}_{2} &  A_{1Q}=j{\sigma}_{1}\otimes I_{2},  & A_{2Q}=j{\sigma}_{1}\otimes j{\sigma}_{3} \\
 A_{3Q}=j{\sigma}_{1}\otimes{\sigma}_{1} & A_{4Q}=j{\sigma}_{1}\otimes {\sigma}_{2} &
\end{array}
\end{eqnarray*}
                                                                                
The TNU-SSD code obtained using this UW code is $\mathbf{S}\left(x_{1},x_{2},x_{3},x_{4}\right)$ given by (\ref{ciodfirstcut}) shown at the top of this page.
\begin{figure*}
{\footnotesize
\begin{equation}
\label{ciodsecondcut}
\left[\begin{array}{cccc}
\alpha(x_{1I}+x_{1Q})+j\alpha(x_{2I}+x_{2Q}) & \alpha(x_{3I}+x_{3Q})+\alpha j(x_{4I}+x_{4Q}) & -\beta(x_{2I}-x_{2Q})+j\beta(x_{1I}-x_{1Q}) & -\beta(x_{4I}-x_{4Q})+j\beta(x_{3I}-x_{3Q}) \\
-\alpha(x_{3I}+x_{3Q})+j\alpha(x_{4I}+x_{4Q}) &  \alpha(x_{1I}+x_{1Q})-j\alpha(x_{2I}+x_{2Q}) & -\beta(x_{4I}-x_{4Q})-j\beta(x_{3I}-x_{3Q}) & \beta(x_{2I}-x_{2Q})+j\beta(x_{1I}-x_{1Q}) \\
\beta(x_{2I}-x_{2Q})-j\beta(x_{1I}-x_{1Q}) & \beta(x_{4I}-x_{4Q})-j\beta(x_{3I}-x_{3Q}) & \alpha(x_{1I}+x_{1Q})+j\alpha(x_{2I}+x_{2Q})&  \alpha(x_{3I}+x_{3Q})+j\alpha(x_{4I}+x_{4Q}) \\
\beta(x_{4I}-x_{4Q})+j\beta(x_{3I}-x_{3Q}) & -\beta(x_{2I}-x_{2Q})-j\beta(x_{1I}-x_{1Q}) &  -\alpha(x_{3I}+x_{3Q})+j\alpha(x_{4I}+x_{4Q})  & \alpha(x_{1I}+x_{1Q})-j\alpha(x_{2I}+x_{2Q})
\end{array} \right] 
\end{equation}  \hrule
}
\end{figure*}
\end{example}
\section{CIODs as a special case of TNU-SSD codes}
\label{sec7}
In this section we give a construction for $2^a\times 2^a$ TNU-SSD codes making use of reducible representations of real Clifford algebras $CA_{2a}$ generated by $2a$ generators. Then,  we show that we can obtain the class of CIODs from the TUN-SSD codes of this construction.
                                                                                
\noindent
{\bf Construction of $2^a\times 2^a$ CIODs :} First we find the irreducible representation of $CA_{2a-1}$. We know that the minimum dimension in which we can get such a representation is $2^{a-1}$. The $2a-1$ anti-Hermitian, anti-commuting $2^{a-1}\times 2^{a-1}$ matrices are explicitly shown below:
\begin{equation}
\label{wtmatricesciod}
\begin{array}{rl}
R(\gamma_{1}) &= j\sigma_3^{\otimes^{a-1}} \\
R(\gamma_{2}) &= I_2^{\otimes^{a-2}} \bigotimes \sigma_1 \\
R(\gamma_{3}) &= I_2^{\otimes^{a-2}} \bigotimes \sigma_2 \\
. & . \\
. & . \\
. & . \\
R(\gamma_{(2k)}) &= I_2^{\otimes^{a-k-1}} \bigotimes \sigma_1 \bigotimes \sigma_3^{\otimes^{k-1}} \\
R(\gamma_{(2k+1)}) &= I_2^{\otimes^{a-k-1}} \bigotimes \sigma_2 \bigotimes \sigma_3^{\otimes^{k-1}} \\
. & . \\
. & . \\
. & . \\
R(\gamma_{(2a-1)}) &=  \sigma_2 \bigotimes \sigma_3^{\otimes^{a-2}} \\
\vspace{0.5cm}
\end{array}
\end{equation}
where $\sigma_1, \sigma_2$ and $\sigma_3$ are given by \eqref{paulimatrices}, $A^{\otimes^{m}} = \underbrace{A\otimes A\otimes A \cdots \otimes A }_{m~~times  }, $
$\gamma_i, ~~ i=1,2,\cdots,(2a-1)$ are the generators of $CA_{2a-1}$ with $\gamma_0=1$ being the identity element and $R(\gamma_0)= I_2^{\otimes^{a-1}}$.\\

Now, define
\begin{equation}
\label{reducible}
A_{iI}=I_2 \otimes R(\gamma_{i-1}), ~~~~ 1 \leq i \leq 2a
\end{equation}
\begin{equation*}
A_{iQ}=A_{1Q}A_{iI}, ~~~~ 1 \leq i \leq 2a
\end{equation*}
where
\begin{equation}
\label{special}
A_{1Q}=\sigma \otimes R(\gamma_0)
\end{equation}
with $\sigma$ being an arbitrary $2\times 2$ unitary and Hermitian matrix.
Now, the code
\begin{equation*}
\label{specialnussd}
S=\sum_{i=1}^{2a}x_{iI}T_{iI}+x_{iQ}T_{iQ}
\end{equation*}
where $T_{iI}, T_{iQ}, 1 \leq i \leq 2a$ are given by \eqref{transform}
is a special case of the codes constructed in Section \ref{sec3}
where the UW-SSD on which we are applying the transform is the one given by \eqref{wtmatricesciod}.
                                                                                
A further special case is the one where we choose $\sigma = \sigma_3$ in \eqref{special} which leads to the class of CIODs as described below. First we define a $4\times 4$ permutation matrix as follows

\begin{equation*}
\label{permutation}
\mathbf{P}\,=\, \left[ \begin{array}{cccc}
1 & 0 & 0 & 0\\
0 & 0 & 1 & 0\\
0 & 0 & 0 & 1\\
0 & 1 & 0 & 0 \end{array} \right].
\end{equation*}
\noindent
Then, we take a pair of variables say $\{x_{i}, x_{i+1}\}.$ These two complex symbols have four real components, $\{x_{iI},x_{iQ},x_{\left(i+1\right)I},x_{\left(i+1\right)I} \}.$ We form another set of four variables from them by applying the defined permutation:
\begin{equation}
\label{transformnew}
{\left[p_{iI} \, p_{iQ}\, p_{\left(i+1\right)I}\, p_{\left(i+1\right)Q}\right]}^{T} \,= \,\mathbf{P}{\left[x_{iI} \,x_{iQ}\,x_{\left(i+1\right)I}\,x_{\left(i+1\right)Q}\right]}^{T}
\end{equation}
 Here superscript $T$ stands for transpose of a matrix. Now for $n=2^{a}$ we have $K=2a$. Hence we choose two consecutive complex variables as a pair and following the above procedure construct the set of $2K$ real variables, ${\{p_{iI},p_{iQ}\}}_{i=1}^{K=2a}$. Then we form a linear STBC with these variables.  The resulting code will be a CIOD in terms of the complex variables ${\{x_{i}\}}_{i=1}^{2a}$. The following example illustrates this.                                   

\begin{figure*}
{\small
\begin{equation}
\label{ciod88temp}
\left[\begin{array}{cccccccc}
 p_{1I}+jp_{2I}& p_{3I}+jp_{4I} & p_{5I}+jp_{6I} &       0        &0&0&0&0\\
-p_{3I}+jp_{4I}& p_{1I}-jp_{2I} &       0        &-p_{5I}-jp_{6I} &0&0&0&0\\
-p_{5I}+jp_{6I}&       0        & p_{1I}-jp_{2I} & p_{3I}+jp_{4I} &0&0&0&0\\
       0       & p_{5I}-jp_{6I} &-p_{3I}+jp_{4I} & p_{1I}+jp_{2I} &0&0&0&0\\

0&0&0&0& p_{1Q}+jp_{2Q} & p_{3Q}+jp_{4Q} & p_{5Q}+jp_{6Q} &       0       \\
0&0&0&0&-p_{3Q}+jp_{4Q} & p_{1Q}-jp_{2Q} &       0        &-p_{5Q}-jp_{6Q} \\
0&0&0&0&-p_{5Q}+jp_{6Q} &       0        & p_{1Q}-jp_{2Q} & p_{3Q}+jp_{4Q}\\
0&0&0&0&       0        & p_{5Q}-jp_{6Q} &-p_{3Q}+jp_{4Q} & p_{1Q}+jp_{2Q} \\
                                                                                
\end{array}\right] 
\end{equation} \hrule
}
\end{figure*}
                                                                                
\begin{figure*}
{\small
\begin{equation}
\label{ciod88}
\left[\begin{array}{cccccccc}

 x_{1I}+jx_{2Q}& x_{3I}+jx_{4Q} & x_{5I}+jx_{6Q} &       0        &0&0&0&0\\
-x_{3I}+jx_{4Q}& x_{1I}-jx_{2Q} &       0        &-x_{5I}-jx_{6Q} &0&0&0&0\\
-x_{5I}+jx_{6Q}&       0        & x_{1I}-jx_{2Q} & x_{3I}+jx_{4Q} &0&0&0&0\\
       0       & x_{5I}-jx_{6Q} &-x_{3I}+jx_{4Q} & x_{1I}+jx_{2Q} &0&0&0&0\\

0&0&0&0& x_{2I}+jx_{1Q} & x_{4I}+jx_{3Q} & x_{6I}+jx_{5Q} &       0       \\
0&0&0&0&-x_{4I}+jx_{3Q} & x_{2I}-jx_{1Q} &       0        &-x_{6I}-jx_{5Q} \\
0&0&0&0&-x_{6I}+jx_{5Q} &       0        & x_{2I}-jx_{1Q} & x_{4I}+jx_{3Q}\\
0&0&0&0&       0        & x_{6I}-jx_{5Q} &-x_{4I}+jx_{3Q} & x_{2I}+jx_{1Q} \\
                                                                                
\end{array}\right] 
\end{equation} \hrule
}
\end{figure*}

\begin{example}
Let
\begin{eqnarray*}
\begin{array}{ccc}

A_{1I}=I_{2}\otimes I_{2}, & A_{2I}=I_{2}\otimes j{\sigma}_{3}, & A_{3I}=I_{2}\otimes{\sigma}_{1}, \\
 A_{4I}=I_{2}\otimes{\sigma}_{2} &  A_{1Q}={\sigma}_{3}\otimes I_{2},  & A_{2Q}={\sigma}_{3}\otimes j{\sigma}_{3} \\
 A_{3Q}={\sigma}_{3}\otimes{\sigma}_{1} & A_{4Q}={\sigma}_{3}\otimes {\sigma}_{2}. &
\end{array}
\end{eqnarray*}
We form a TNU-SSD code by setting $\alpha = \beta =\frac{1}{2}$ leading to the following code:
                                                                                
\begin{equation}
\label{specialnussdnew}
S=\sum_{i=1}^{2a}p_{iI}T_{iI}+p_{iQ}T_{iQ}
\end{equation}

\begin{equation*}
=\left[ \begin{array}{cccc}
p_{1I}+jp_{2I} & p_{3I}+jp_{4I} & 0 & 0\\
-p_{3I}+jp_{4I} & p_{1I}-jp_{2I} & 0 & 0 \\
0 & 0 & p_{1Q}+jp_{2Q} & p_{3Q}+jp_{4Q} \\
0 & 0 & -p_{3Q}+jp_{4Q} & p_{1Q}-jp_{2Q}
\end{array} \right]
\end{equation*}
                                                                                
Now using (\ref{transformnew}) the code becomes
\begin{equation*}
\label{ciod44}
\left[ \begin{array}{cccc}
x_{1I}+jx_{2Q} & x_{3I}+jx_{4Q} & 0 & 0\\
-x_{3I}+jx_{4Q} & x_{1I}-jx_{2Q} & 0 & 0 \\
0 & 0 & x_{2I}+jx_{1Q} & x_{4I}+jx_{3Q} \\
0 & 0 & -x_{4I}+jx_{3Q} & x_{2I}-jx_{1Q}
\end{array} \right]
\end{equation*}
which is a $4\times 4$ CIOD.
\end{example}
\begin{example}
Here we form an $8\times 8$ CIOD code following our approach,
\begin{eqnarray*}
\begin{array}{ccc}
A_{1I}=I_{2}\otimes I_{2}\otimes I_{2}, & A_{2I}=jI_{2}\otimes {\sigma}_{3}\otimes {\sigma}_{3} & \\
A_{3I}=I_{2}\otimes I_{2}\otimes{\sigma}_{1}, & A_{4I}=I_{2}\otimes I_{2}\otimes{\sigma}_{2} & \\
A_{5I}=I_{2}\otimes {\sigma}_{1}\otimes {\sigma}_{3} , & A_{6I}=I_{2}\otimes {\sigma}_{2}\otimes {\sigma}_{3}\\
A_{1Q}={\sigma}_{3}\otimes I_{2}\otimes I_{2}, & A_{iQ}=A_{1Q}A_{iI},~~i=2,3,\cdots,6.
\end{array}
\end{eqnarray*}
Now we form the code as in \eqref{specialnussdnew}, with $\alpha=\beta=\frac{1}{2}$, and the transformed variables ${\{p_{iI}, p_{iQ}\}}_{i=1}^{6}$ as shown in \eqref{ciod88temp} at the top of the next page. Now using \eqref{transformnew} we get \eqref{ciod88} as shown at the top of the next page which is the same as the $8\times 8$ CIOD.
\end{example}
                                                                                
\begin{remark}
It is interesting to observe that \eqref{reducible} represents an reducible representation and this construction based on reducible representation leads to NU-SSD codes and CIODs. We are not aware of any other code constructions that make use of reducible representations of groups.
\end{remark}

\section{Discussion}
\label{sec8}
In the most general form a STBC is simply a finite set of matrices with complex entries. One way of obtaining a STBC is by first specifying a design as in \eqref{ldceqn} and then let the variables $\{x_i \}_{i=1}^K$ take values from a finite set of complex numbers like $M$-ary PSK and QAM. Notice that two different designs taking values from two different signal sets may result in the same STBC (finite set of complex matrices). It is important to notice that the attribute of single-symbol decodability is that of the design and not that of the resulting STBC when a signal set is specified for the variables. We will explain this by an example: Note that,

\begin{equation*}
\left[\begin{array}{cc}
x_{1I}-jx_{2Q} & x_{2I}+jx_{1Q}\\
-x_{2I}-jx_{1Q} & x_{1I}-jx_{2Q}\end{array}\right]
\end{equation*}
is a linear design which is  a $2\times 2$ SSD code. If $\mathcal{A}$ is a finite subset of the complex field from which the variables take values from, then the resulting STBC is the set of matrices
\begin{equation}
\label{original}
\{x_{1I}A_{1I}+x_{1Q}A_{1Q}+x_{2I}A_{2I}+x_{2Q}A_{2Q}\} 
\end{equation}
where $ x_{iI}+jx_{iQ} \in \mathcal{A},i=1,2 $ and 
\begin{eqnarray*}
A_{1I}=\left[\begin{array}{cc}
1 & 0\\
0 & 1 \end{array}\right], 
A_{1Q}=\left[\begin{array}{cc}
0 & j\\
-j & 0 \end{array}\right]\\
A_{2I}=\left[\begin{array}{cc}
0 & 1\\
-1 & 0 \end{array}\right], 
A_{2Q}=\left[\begin{array}{cc}
-j & 0\\
0 & -j \end{array}\right].
\end{eqnarray*}
Now consider a $2\times 2$ real non-singular linear transform matrix,
\begin{equation*}
T=\left[\begin{array}{cc}
cos(\theta) & sin(\theta)\\
-sin(\theta) & cos(\theta)\end{array}\right] \textrm{and} 
\left[\begin{array}{c}
y_{iI}\\
y_{iQ}
\end{array}\right]=T^{-1}\left[\begin{array}{c}
x_{iI}\\
x_{iQ}
\end{array}\right]
\end{equation*}
so that corresponding to every point $x_{iI}+jx_{iQ} \in \mathcal{A}$ there is one and only one point $y_{iI}+jy_{iQ} \in \mathcal{\widetilde{A}}.$ Now the set of codeword matrices in \eqref{original} can be also be written as,
\begin{equation}
\label{modified}
\{y_{1I}\widetilde{A}_{1I}+y_{1Q}\widetilde{A}_{1Q}+y_{2I}\widetilde{A}_{2I}+y_{2Q}\widetilde{A}_{2Q} \}
\end{equation}
where $y_{iI}+jy_{iQ} \in \mathcal{\widetilde{A}},i=1,2$ 
and 
\begin{equation*}
\begin{array}{c}
\widetilde{A}_{1I}=cos(\theta) {A}_{1I}-sin(\theta) {A}_{1Q}\\
\widetilde{A}_{1Q}=sin(\theta) {A}_{1I}+cos(\theta) {A}_{1Q}\\
\widetilde{A}_{2I}=cos(\theta) {A}_{2I}-sin(\theta) {A}_{2Q}\\
\widetilde{A}_{2Q}=sin(\theta) {A}_{2I}+cos(\theta) {A}_{2Q}
\end{array}
\end{equation*}
It is obvious that \eqref{original} and \eqref{modified} represent the same STBC but in \eqref{original} the weight matrices are unitary but for any non-zero $\theta$ the weight matrices in \eqref{modified} are not unitary. 

In \cite{WWX1,WWX2}, the authors start from a QOD and taking appropriate transformation of the variables of the design obtain UW-SSD designs which intersect with the YGT codes. Further transformations are employed to maximize the coding gain which result in NUW-SSDs. It is an interesting open problem to identify the transformations which result in the classes of NUW-SSD codes obtained in this paper. 

Another important direction for further research is to settle the conjecture regarding the maximum diversity product of CUW-SSD codes: For a given signal set the diversity product expression \eqref{codgain} is maximum when $2m=n$, i.e., when the discriminant of the code is traceless. 

Another important observation which opens up further investigation is the following. The choice of  $\sigma_3$ in \eqref{special} is responsible for the codes of the construction resulting in CIODs. The construction will continue to work leading to codes with different structures for different choice of a $2\times 2$ matrix as long as it is Hermitian.
\appendices
\section{Quadratic Spaces and Clifford algebras}
\label{Append1}
\subsection{Quadratic Spaces and Clifford algebras}
In this subsection we briefly describe the notion of quadratic forms and Clifford algebras along with their basic structural results needed for our purposes. The proofs and further results concerning quadratic forms can be found in \cite{Lam} and \cite{Ome} and concerning Clifford algebras can be found in \cite{Por}, \cite{Dix} and \cite{Han}.
                                                                                
Let $V$ be a finite-dimensional vector space over the real field $\mathbb{R}$. A {\it quadratic form} (QF) on $V$ is a mapping $Q:V \rightarrow \mathbb{R}$ such that \\
\begin{tabular}{rl}
(i) & $Q(\alpha v)  = \alpha^2Q(v), ~~~ \alpha \in \mathbb{R},~~v \in V$ \\
(ii) & the associated form \\
 & $B(v,w)  =\frac{1}{2}\left\{Q(v)+Q(w)-Q(v-w)  \right\},~ v,w \in V$ \\
& is bilinear.\\
\end{tabular}
When such a QF exists, the pair $(V,Q)$ is said to be a {\it quadratic space}. Note that every vector space over ${\mathbb R}$ becomes a quadratic space with respect to the trivial quadratic form $Q(v)=0,~~ \forall ~~ v\in V$.\\
                                                                                
Let $p,q$ be non-negative integers with $p+q=n > 0$ and define the quadratic form on ${\mathbb R}^{p+q}$ by
\begin{equation*}
\label{gqf}
Q_{p,q}(u)=-( u_1^2+\cdots +u_p^2)+(u_{p+1}^2+\cdots+u_{p+q}^2),
\end{equation*}
for $u=(u_1,\cdots, u_{p+q});$ the resulting read quadratic space is called $(p,q)$-{\it Minkowski space} and we denote it by $({\mathbb R}^{p,q},Q_{p,q})$. Clearly, ${\mathbb R}^{n,0},Q_{n,0}$ reduces to $({\mathbb R}^n, -|.|^2)$ and ${\mathbb R}^{0,n},Q_{0,n}$ reduces to $({\mathbb R}^n, |.|^2)$ where $|.|$ is the Euclidean norm given by
\begin{equation*}
|u|^2= (u_1^2+u_2^2+\cdots+u_n^2)
\end{equation*}
Now, let $(V,Q)$ be an arbitrary quadratic space and $e_i$ a basis of $V$. Then
\begin{equation*}
B(u,v)=Q(v)= \sum_{i,j}B(e_i,e_j)v_iv_j, ~~~~ v=\sum_{i}v_ie_i,
\end{equation*}
and if there a basis which is {\it B-orthogonal} in the sense that
\begin{equation*}
B(e_i,e_j)=0,~~~ i \neq j,
\end{equation*}
the expression for $Q(v)$ reduces to the diagonal form $Q(v)=\sum_{i}Q(e_i)v_i^2$. Such a basis is easily constructable. The subset of $V$ given by
\begin{equation*}
Rad(V,Q)=\left\{w \in V| B(u,w)=0, ~~ \forall ~~ u\in V \right\}= V^\perp
\end{equation*}
is called the {\it radical of} $(V,Q)$. The quadratic space $(V,Q)$ is said to be {\it non-degenerate } if $Rad(V,Q)=\{ 0 \};$ otherwise it is said to be {\it degenerate}. The space $V$ can be written as the $B$-orthogonal direct sum
                                                                                
\begin{equation}
\label{decomposition}
V=Rad(V,Q) \oplus Rad(V,Q)^\perp
\end{equation}
of $Rad(V,Q)$ and its $B$-orthogonal complement.
\begin{lemma}
Let $(V,Q)$ be a quadratic space with $B$-orthogonal decomposition as in \eqref{decomposition}. Then,\\
\begin{tabular}{rl}
(a) & $Q==0$  on $Rad(V,Q)$.\\
(b) & $Rad(V,Q)\perp, Q$ is isomorphic to ${\mathbb R}^{p,q}$ where $p,q$ \\
    & depend only on $Q$.
\end{tabular}
\end{lemma}
\begin{definition}
Let ${\mathbf A}$ be an associative algebra over the field ${\mathbb R}$ with identity ${\mathbf 1}$ and $\gamma: V \rightarrow {\mathbf A}$ an ${\mathbb R}$-linear embedding of $V$ into ${\mathbf A}$. The pair $({\mathbf A}, \gamma)$ is said to be a real Clifford algebra for $(V,Q)$ when
\begin{enumerate}
\item ${\mathbf A}$ is generated as an algebra by $\{\gamma(v): v \in V \} \cup \{\lambda {\mathbf 1}: \lambda \in {\mathbb R} \}$,
\item $(\gamma(v))^2= -Q(v) {\mathbf 1}, ~~~ \forall ~~ v \in V$.
\end{enumerate}
\end{definition}
The second condition in the above definition ensures that ${\mathbf A}$ is an algebra in which there exists a ``square root'' of the quadratic form $-Q$.
                                                                                
\begin{definition}
The {\it Pauli matrices} in ${\mathbb C}^{2\times 2}$ are
\begin{eqnarray*}
\lambda_0=
\left[
\begin{array}{rr}
1 & 0 \\
0 &1
\end{array}
\right],~~
\lambda_1=
\left[
\begin{array}{rr}
1 & 0 \\
0 & -1
\end{array}
\right],~~ \\ \nonumber
\lambda_2=
\left[
\begin{array}{rr}
0 & -j \\
j &0
\end{array}
\right],~~
\lambda_3=
\left[
\begin{array}{rr}
0 & 1 \\
1 &0
\end{array}
\right],~~
\end{eqnarray*}
and the {\it associated Pauli matrices} are
                                                                                
\begin{eqnarray*}
\mu_0=
\left[
\begin{array}{rr}
1 & 0 \\
0 &1
\end{array}
\right],~~
\mu_1=
\left[
\begin{array}{rr}
j & 0 \\
0 & -j
\end{array}
\right],~~ \\ \nonumber
\mu_2=
\left[
\begin{array}{rr}
0 & 1 \\
-1 &0
\end{array}
\right],~~
\mu_3=
\left[
\begin{array}{rr}
0 & j \\
j &0
\end{array}
\right].
\end{eqnarray*}
\end{definition}
It is easily seen that $\lambda_1^2=\lambda_2^2=\lambda_3^2=I$ and $\mu_1^2=\mu_2^2=\mu_3^2=-I.$ Moreover, $\lambda_j \lambda_k = -j\lambda_l$ and $\mu_j \mu_k =\mu_l$ when $\{j,k,l \}$ is a cyclic permutation of ${1,2,3}$.  Pauli matrices and their associates occur throughout the theory of Clifford algebras. For instance, let,
\begin{eqnarray*}
\begin{array}{rl}
{\mathbb U}_{0,0} &  =\{a \sigma_0: a \in {\mathbb R} \}, \\
{\mathbb U}_{1,0} &  =
\left\{ \left[
\begin{array}{rl}
x & y \\
y & x
\end{array}
\right] : x,y \in {\mathbb R} \right\}, \\
 {\mathbb U}_{0,1} & =
\left\{ \left[
\begin{array}{rl}
x & y \\
-y & x
\end{array}
\right] : x,y \in {\mathbb R} \right\}, \\
{\mathbb U}_{0,2} & =
\left\{ \left[
\begin{array}{rl}
x_0+jx_1 & x_2+jx_3 \\
-x_2+jx_3 & x_0-jx_1
\end{array}
\right] : x_i \in {\mathbb R} \right\}, \\
 & =
\left\{ \left[
\begin{array}{rl}
z_1 & z_2 \\
-z_2^* & z_1^*
\end{array}
\right] : z_i \in {\mathbb C} \right\}.
\end{array}
\end{eqnarray*}
                                                                                
Each of the above is an associative subalgebra (over ${\mathbb R}$) of
${\mathbb C}^{2 \times 2}$ having an identity element, and
\begin{equation*}
{\mathbb U}_{0,0} \equiv {\mathbb R}, ~~
{\mathbb U}_{1,0} \equiv {\mathbb R} \oplus {\mathbb R}, ~~
{\mathbb U}_{0,1} \equiv {\mathbb C}, ~~
{\mathbb U}_{0,2} \equiv {\mathbb H}
\end{equation*}
where ${\mathbb H}$ is the Hamilton's algebra of quaternions. As the notation suggests, ${\mathbb U}_{p,q}$ also is a Clifford algebra for ${\mathbb R}^{p,q}$ with respective embeddings $\gamma$ given by
\begin{equation*}
0 \rightarrow 0,~~~ y \rightarrow y\lambda_3, ~~~ y \rightarrow y \mu_2, ~~~ (x_1,x_2) \rightarrow x_1\mu_1+x_2\mu_2.
\end{equation*}
\begin{definition}
We will call the real Clifford algebras for $({\mathbb R}^{1,n-1}, Q_{1,n-1})$ to be {\it Minkowski Clifford algebras} and those for $({\mathbb R}^{0,n}, Q_{0,n})$ to be {\it Euclidean Clifford algebras}.
\end{definition}

\section{CUW-SSD codes from Minkowski Clifford algebras}
\label{Append2}
In this section we describe another construction of UW-SSD codes based on representations of Minkowski Clifford algebras.
\begin{theorem}
\label{minkowskithm}
The  $n \times n$ UW code given by
\begin{equation}
\label{uwssdalter}
S=\sum_{i=1}^{K}\left(x_{iI}A_{iI}+x_{iQ}A_{iQ}\right) \end{equation}
is UW-SSD  if there exists a matrix ${\hat A}_{1Q}$ satisfying the following interrelationships with the weight matrices:
\begin{eqnarray}
\label{suffcondalter}
\begin{array}{rl}
A_{iI}^{H}  & =-A_{iI}\\
 A_{iI}A_{jI}  & = - A_{jI}A_{iI}, \quad   1\leq i\neq j\leq K \\
{\hat A}_{1Q}^{H}   &={\hat A}_{1Q} \\
{\hat A}_{1Q}A_{iI} & =-  A_{iI}{\hat A}_{1Q} ~~~~ 1 \leq i \leq K \\
A_{iQ} &={\hat A}_{1Q}A_{iI} ~~~~ 1 \leq i \leq K.
\end{array}
\end{eqnarray}
(Note that ${\hat A}_{1Q}$ is only an intermediate matrix using which the set of matrices $A_{iQ},~~i=1,2,\cdots,K$ are defined.) 
\end{theorem}
\begin{proof}
The proof is by direct verification of \eqref{A2} for the weight matrices of the code.
\begin{eqnarray*}
\begin{array}{rl}  A_{iI}^{H}A_{jQ}+A_{jQ}^{H}A_{iI} & = A_{iI}^{H}{\hat A}_{1Q}A_{jI}+A_{jI}^{H}{\hat A}_{1Q}^{H}A_{iI}  \\
&  = A_{iI}{\hat A}_{1Q}A_{jI}+A_{jI}{\hat A}_{1Q}A_{iI} \\
&  = (A_{iI}^HA_{jI}+A_{jI}^HA_{iI}){\hat A}_{1Q}\\
&  = 0({\hat A}_{1Q})=0.
\end{array}
\end{eqnarray*} This shows that \eqref{A21} is satisfied. Next, we show that  \eqref{A22} is also satisfied:
\begin{eqnarray*}
\begin{array}{rl}
A_{iI}^HA_{jI}+A_{jI}^HA_{iI} & = -(A_{iI}A_{jI}+A_{jI}A_{iI})=0.\\
\end{array}
\end{eqnarray*}
To prove that  ({\ref{A23}}):
\begin{eqnarray*}
\begin{array}{rl} A_{iQ}^HA_{jQ}+A_{jQ}^HA_{iQ} & = A_{iI}^H{\hat A}_{1Q}^H{\hat A}_{1Q}A_{jI}+A_{jI}^H{\hat A}_{1Q}^H{\hat A}_{1Q}A_{iI}\\
& =A_{iI}^HA_{jI}+A_{jI}^HA_{iI} \\
& = -(A_{iI}A_{jI}+A_{jI}A_{iI})=0.\\
\end{array}
\end{eqnarray*}
This shows that the code \eqref{minkowskithm} is UW-SSD.
\end{proof}
\begin{theorem}
\label{cuwssdalter}
Consider the following $2^a\times 2^a$  weight matrices
\begin{eqnarray}
\label{wtmatricesalter}
\begin{array}{rl}
A_{1I} & = j \sigma_3^{\otimes^{a}} \\
                                                                                
A_{iI} & = R(\gamma_{i+1}), ~~~ 2 \leq i \leq 2a, \\
\mbox{   and    }~~~A_{iQ} & ={\hat A}_{1Q}A_{iI}, ~~~ 1\leq i \leq K   \\
                                                                                
\mbox{   where  } ~~~ {\hat A}_{1Q} & = I_2^{\otimes^{a-1}} \bigotimes j\sigma_1\end{array}
\end{eqnarray}
and $\sigma_1, \sigma_2$ and $\sigma_3$ are given by \eqref{paulimatrices}.
 With these weight matrices the resulting  $2^a \times 2^a$ code $S(x_1,x_2,\cdots,x_{2a})$ given by \eqref{cuwssdeqalter} at the top of the next page,  where
\begin{figure*}
\begin{equation}
\label{cuwssdeqalter}
\begin{array}{l}
 I_2^{\otimes^{a-1}} \bigotimes \rho_{x_2}+j \sigma_3^{\otimes^{a-1}} \bigotimes \sigma_{x_1}
  + \sum_{i=2}^{a}
\left[
 I_2^{\otimes^{a-i}}  \bigotimes \sigma_1 \bigotimes \sigma_3^{\otimes^{i-2}} \bigotimes \sigma_{x_{2i-1}}+ I_2^{\otimes^{a-i}} \bigotimes \sigma_2 \bigotimes \sigma_3^{\otimes^{i-2}} \bigotimes \sigma_{x_{2i}}
\right] \\
\end{array}
\end{equation} \\ \hrule
\end{figure*}
\[\begin{array}{rl}
x_i &= x_{iI}+jx_{iQ} \\
                                                                                
\sigma_{x_i} &= \left[
\begin{array}{rr}
x_{iI} & -jx_{iQ} \\
-jx_{iQ} & -x_{iI}
\end{array}
\right]  \mbox{  and }\\
                                                                                
\rho_{x_i} &= \left[
\begin{array}{rr}
-x_{iQ} & jx_{iI} \\
jx_{iI} & x_{iQ}
\end{array}
\right]
\end{array}
\]
is a UW-SSD code in $2a$ complex variables with rate ($\frac{a}{2^{a-1}}$).
\end{theorem}
\begin{proof}
From the representation matrices of Lemma~\ref{lemma2} and by the construction of weight matrices it is easily checked by direct verification that all the sufficient conditions of Theorem \ref{minkowskithm} for an UW-SSD are satisfied.
\end{proof}
It can be verified by direct computation that the set of weight matrices given by \eqref{wtmatricesalter} constitute a $2^a \times 2^a$ matrix representation of the Clifford algebra ${\mathbb U}_{1,2a-1}$ and the set of weight matrices given by \eqref{wtmatrices} constitute a $2^a\times 2^a$ matrix representation of the Clifford algebra ${\mathbb U}_{0,2a}$. The quadratic space associated with the ${\mathbb U}_{0,2a-1}$ is the Minkowski space and the quadratic space associated with the ${\mathbb U}_{0,2a}$ being the Euclidean space. To highlight this difference we associate the name Minkowski to the codes given by \eqref{cuwssdeqalter} as follows:
\begin{definition}
The $2^a\times 2^a$ STBCs given by \eqref{cuwssdeqalter} are defined to be $a-$Minkowski-Clifford Unitary Weight SSD (MCUW-SSD) codes.
\end{definition}
The $1-$MCUW-SSD code is
\[
\begin{array}{c}
S(x_1,x_2)=j\sigma_{x_1}+\rho_{x_2} = \left[
\begin{array}{rr}
-x_{2Q}-jx_{1I}  & x_{1Q}+jx_{2I} \\
x_{1Q}-jx_{2I} & x_{2Q}-jx_{1I}
\end{array}
\right]
\end{array}
\]
and the $2-$MCUW-SSD code is
\begin{eqnarray*}
\begin{array}{l}
S(x_1,x_2,x_3,x_4) \\
=j\sigma_{x_3}\bigotimes \sigma_{x_1} + I_2 \bigotimes \rho_{x_2} +\sigma_{1} \bigotimes \sigma_{x_3} +\sigma_{2}\bigotimes \sigma_{x_4}\\
=j\sigma_{x_3}\bigotimes
\left[
\begin{array}{rr}
x_{1I}   & -jx_{1Q}\\
-jx_{1Q} & -x_{1I}
\end{array}
\right]
 + I_2 \bigotimes
\left[ \begin{array}{rr}
-x_{2Q}   & -jx_{2I}\\
-jx_{2I} & x_{2Q}
\end{array}
\right] \\ \nonumber
+\sigma_{1} \bigotimes
\left[
\begin{array}{rr}
x_{3I}   & -jx_{3Q}\\
-jx_{3Q} & -x_{3I}
\end{array}
\right]
+\sigma_{2}\bigotimes
\left[
\begin{array}{rr}
x_{4I}   & -jx_{4Q}\\
-jx_{4Q} & -x_{4I}
\end{array}
\right]. \nonumber
\end{array}
\end{eqnarray*}
                                                                                
which is
                                                                                
{\small
                                                                                
\begin{eqnarray*}
\left[
\begin{array}{rrrr}
jx_{1I}-x_{2Q}  & x_{1Q}+jx_{2I} & x_{3I}+jx_{4I}& x_{4Q}-jx_{3Q}  \\
x_{1Q}+jx_{2I} & x_{2Q}-jx_{1I} & x_{4Q}-jx_{3Q} & -x_{3I}-jx_{4I} \\ -x_{3I}+jx_{4I} & x_{4Q}+jx_{3Q} & -x_{2Q}-jx_{1I} & -x_{1Q}+jx_{2I} \\
x_{4Q}+jx_{3Q} & x_{3I}-jx_{4I} & -x_{1Q}+jx_{2I} & x_{2Q}+jx_{1I}
\end{array}
\right].
\end{eqnarray*}
}
\subsection{Normalized MCUW-SSD codes}
In this subsection we show that if normalization is carried out on the MCUW-SSD codes then it turns out to be the same as the CUW-SSD codes.\\
\begin{theorem}
\label{normalMisC}
The normalized version of the code \eqref{uwssdalter} satisfying the conditions of \eqref{suffcondalter} satisfy the conditions given by \eqref{suffcondssd}.
\end{theorem}
\begin{proof}
Let
\begin{equation*}
\label{uwssdaltern} S=\sum_{i=1}^{k}\left(x_{iI}{\tilde A}_{iI}+x_{iQ}{\tilde A}_{iQ}\right)
\end{equation*}
be the normalized version of \eqref{uwssdalter} obtained by the substitution
\begin{eqnarray*}
{\tilde A}_{iI}=A_{1I}^HA_{iI} \\
{\tilde A}_{iQ}=A_{1I}^HA_{iQ}.
\end{eqnarray*}
Notice that ${\tilde A}_{1I}=I$ and
\begin{eqnarray*}
{\tilde A}_{1Q}=A_{1I}^HA_{1Q}=A_{1I}^H{\hat A}_{1Q}A_{1I}=-A_{1I}^HA_{1I}{\hat A}_{1Q}=-{\hat A}_{1Q}.
\end{eqnarray*}

For $2 \leq i \leq K$, We have
\begin{eqnarray}
\label{cond1}
\begin{array}{rl}
{\tilde A}_{iI}^H & = (A_{1I}^HA_{iI})^H= A_{iI}^HA_{1H} = -A_{iI}A_{1I}\\
 & =A_{1I}A_{iI}=-A_{1I}^HA_{iI} = -{\tilde A}_{iI}.
\end{array}
\end{eqnarray}
                                                                                
For $2 \leq i \neq j \leq K$, we have
\begin{eqnarray}
\label{cond2}
\begin{array}{rl}
{\tilde A}_{iI} {\tilde A}_{jI} & = A_{1I}^HA_{iI}A_{1I}^HA_{jI}= A_{1I}A_{iI}A_{1I}A_{jI} \\
 & = -A_{1I}A_{jI}A_{1I}A_{iI} = -{\tilde A}_{jI}{\tilde A}_{iI}.
\end{array}
\end{eqnarray}
                                                                                
Also, we have,
\begin{eqnarray}
\label{cond3}
\begin{array}{rl}
{\tilde A}_{1Q}^H & = (A_{1I}^HA_{1Q})^H = A_{1Q}^HA_{1I} = A_{1I}^HA_{1Q}^HA_{1I} \\
 & = A_{1I}^HA_{1Q}^HA_{1I} = A_{1I}^HA_{1Q}= {\tilde A}_{1Q}^H.
\end{array}
\end{eqnarray}

For $2 \leq i \leq K$, we have
\begin{eqnarray}
\label{cond4}
\begin{array}{rl}
{\tilde A}_{iQ} & = A_{1I}^HA_{iQ} = A_{1I}^H{\hat A}_{1Q}A_{iI} \\
 & = -A{1I}{\hat A}_{1Q}A_{iI} = {\hat A}_{1Q}A_{1I}A_{iI} = -{\hat A}_{1Q}A_{1I}^HA_{iI} \\
 & = -{\hat A}_{1Q}{\tilde A}_{iI} = {\tilde A}_{1Q}{\tilde A}_{iI}.
\end{array}
\end{eqnarray}
                                                                                
We proceed to show that ${\tilde A}_{1Q}$ commutes with all ${\tilde A}_{iI}$ for $2 \leq i$:
\begin{eqnarray}
\label{cond5}
\begin{array}{rl}
{\tilde A}_{1Q}{\tilde A}_{iI} & = -{\hat A}^{1Q}A_{1I}^HA_{iI}= {\hat A}_{1Q}A_{1I}A_{iI} \\
 & = -A_{1I}{\hat A}_{1Q}A_{iI} = A_{1I}A_{iI}{\hat A}_{1Q} = -A_{1I}^HA_{iI}{\hat A}_{1Q} \\
 & = (A_{1I}^HA_{iI})(-{\hat A}_{1Q}) = {\tilde A}_{iI}{\tilde A}_{1Q}.
\end{array}
\end{eqnarray}
                                                                                
The equations \eqref{cond1}-\eqref{cond5} show that all the conditions of \eqref{suffcondssd} are fulfilled.
\end{proof}
\section{Proof for the claim $k \neq (2a+1)$ in Theorem \ref{ubthm}}
\label{Append3}
\begin{proof}
The proof is by contradiction- suppose $K=2a+1$ in \eqref{doublenormal}.

By putting $j=1$ and $i\geq 2$ in \eqref{A21} we get,
\begin{equation*}
\begin{array}{c}
A_{iI}^{H}A_{1Q}+A_{1Q}^{H}A_{iI}=0, \quad  2\leq i\leq 2a+1 \\
\Rightarrow \quad A_{iI}A_{1Q} = A_{1Q}^{H}A_{iI}, \quad  2 \leq i\leq 2a+1
\end{array}
\end{equation*}
\begin{equation}
\label{a15}
\Rightarrow \quad A_{1Q}A_{iI} = A_{iI}A_{1Q}^{H} \quad \textrm{for} \, 2\leq i\leq 2a+1.
\end{equation}
Define
\begin{equation}
\label{a17}
H=\frac{A_{1Q}^{H}+A_{1Q}}{2}, \quad   S=\frac{A_{1Q}-A_{1Q}^{H}}{2}
\end{equation}
so that 
\begin{equation*}
\label{a18}
 A_{1Q}= H+S, \quad H=H^{H} \quad \textrm{ and }   S^{H}=-S.
\end{equation*}
Now using \eqref{a15} we have,
\begin{equation}
\label{a19}
\left(H+S\right)A_{iI} = A_{iI}\left(H-S\right), \quad 2 \leq i \leq 2a+1.
\end{equation}
Also, let us define $P_i$ and $Q_i$ by
\begin{equation}
\label{a20}
HA_{iI} = A_{iI}P_{i}; \quad SA_{iI} = A_{iI}Q_{i} \quad  2 \leq i \leq 2a+1.
\end{equation}
Note that  $P_{i}$ is Hermitian and $Q_{i}$ is anti-Hermitian. Now from \eqref{a19} and \eqref{a20} we get,
\begin{equation*}
A_{iI}\left(P_{i}+Q_{i}\right)=A_{iI}\left(H-S\right)
\end{equation*}
which implies 
\begin{equation}
\label{a21}
 P_{i}+Q_{i}=H-S \quad 2 \leq i \leq  2a+1
\end{equation}
taking Hermitian of both sides of which we get,
\begin{equation}
\label{a22}
P_{i}-Q_{i}=H+S  \quad 2 \leq i \leq  2a+1.
\end{equation}
Adding and subtracting \eqref{a21} and \eqref{a22} we get,
\begin{equation*}
\label{a23}
P_{i}=H, \quad Q_{i}=-S  \quad 2 \leq i \leq 2a+1
\end{equation*}
which implies
\begin{equation}
\label{a24}
HA_{iI}=A_{iI}H , \quad SA_{iI}=-A_{iI}S  \quad 2 \leq i \leq 2a+1.
\end{equation}
Now from \eqref{a24}, we see that the Hermitian matrix $H$ commutes with all $A_{iI}$'s for $i= 2 \cdots 2a+1$. But this set of $A_{iI}'s$ represents an irreducible representation of $CA_{2a}$, and hence from Lemma~\eqref{cliff_group} and Schur's lemma, $H=\alpha I_n$.
Now from \eqref{a17} we have $A_{1Q}=\alpha I_n+S$ and since  $A_{1Q}$ is unitary, we have
\begin{eqnarray*} 
(\alpha I_n+S) (\alpha I_n+S)^H = I_n \\
\Rightarrow \, \left(\alpha I_{n\times n}+S\right) \left(\alpha I_{n\times n}-S\right)=I_n \\
\Rightarrow \, \left({\alpha}^{2} I_{n\times n}-S^{2}\right)=I_n \\
\Rightarrow \, S^{2}=-\left(1-{\alpha}^2\right)I_n \\
 \Rightarrow \, S= {\beta}_{1} M,
\end{eqnarray*}
where $M$ is an unitary skew-Hermitian matrix with $\beta_1 \in \mathbb{R}$ and from \eqref{a24}
\begin{equation*}
\label{a25}
MA_{iI}=-A_{iI}M, \quad  2 \leq i \leq 2a+1.
\end{equation*}
Note that $\alpha \leq 1$, because otherwise $S$ can not be a skew-Hermitian matrix.\\

A nondegenerate irreducible representation of $CA_{2a+1}$ can be found from an irreducible representation of $CA_{2a}$ in the following way (See Proposition A.6 of \cite{TiH}) : If $ {\cal R}_{2a}(\gamma_i), ~~i=2,3, \cdots 2a+1$, is a representation of the generators of $CA_{2a}$, then a representation of the generators of $CA_{2a+1}$ can be generated as the set $ {\cal R}_{2a+1}(\gamma_i)={\cal R}_{2a}(\gamma_i),~ i=2,3, \cdots 2a+1$ and ${\cal R}_{2a+1}(\gamma_{2a+1})= \pm j {\prod}_{i=2}^{2a}{\cal R}_{2a+1}(\gamma_{i})$. Using this fact, if the set ${\{A_{iI}\}}_{i=2}^{2a+1}$ is an irreducible representation of $CA_{2a}$ then the set $\{\pm j {\prod}_{i=2}^{2a+1}A_{iI}\}\bigcup{\{A_{iI}\}}_{i=2}^{2a+1}$ is an irreducible representation of $CA_{2a+1}$. Now, we have the two sets $ \{M\}\bigcup {\{A_{iI}\}}_{i=2}^{2a+1}$ and $ \{A_{2a+2}\}\bigcup {\{A_{iI}\}}_{i=2}^{2a+1}$ where, 
\begin{equation*}
\label{a26}
A_{2a+2}=j{\prod}_{i=2}^{2a+1}A_{iI}
\end{equation*}
constituting two irreducible representations of $CA_{2a+1}$. Hence from Lemma~ \ref{unique_last} we have $M=\pm A_{2a+2} $ and  also,
\begin{equation*}
\label{a29}
A_{1Q}=(\alpha I_n+\beta A_{2a+2}), \, \textrm{where} \,\beta=\pm {\beta}_{1} \, \in \, \mathbb{R}.
\end{equation*}
Now, $A_{1Q}A_{1Q}^{H}=I_n 
\Rightarrow \, \left(\alpha I_n+\beta A_{2a+2}\right)\left(\alpha I_n-\beta A_{2a+2}\right)=I_n 
\Rightarrow \, \left({\alpha}^{2}+{\beta}^{2}\right)I_n=I_n $ leads to
\begin{equation*}
\label{a30}
{\alpha}^{2}+{\beta}^{2}=1.
\end{equation*}

Putting $i=1$ and $j\geq 2$ in  \eqref{A21}  we get,
\begin{eqnarray}
\begin{array}{l}
A_{jQ}+A_{jQ}^{H}=0 \quad \textrm{  for  }  2\leq j\leq 2a+1\\ 
\label{a16}
\Rightarrow  A_{1Q}{A^{\prime}}_{jQ}=-{A^{\prime}}_{jQ}^{H}A_{1Q}^{H}\quad \mbox{ since }  A_{jQ}=A_{1Q}{A^{\prime}}_{jQ}.
\end{array}
\end{eqnarray}

Now, from \eqref{a16}, for    $2 \leq j \leq 2a+1 $
\begin{eqnarray*}
\begin{array}{l}
A_{1Q}{A^{\prime}}_{jQ}=-{A^{\prime}}_{jQ}^{H}A_{1Q}^{H}\\
\Rightarrow (\alpha I_n+\beta A_{2a+2}){A^{\prime}}_{jQ}=\alpha {A^{\prime}}_{jQ}-\beta {A^{\prime}}_{jQ}A_{2a+2}.
\end{array}
\end{eqnarray*}
Therefore,
\begin{equation*}
\label{a31}
A_{2a+2}{A^{\prime}}_{jQ}=-{A^{\prime}}_{jQ}A_{2a+2}, \quad  2 \leq j \leq 2a+1.
\end{equation*}

Now the set $\{A_{2a+2}\}\bigcup{\{{A^{\prime}}_{iQ}\}}_{i=2}^{2a+1} $ satisfies all the conditions required for it to be a faithful representation of the generators of the $CA_{2a+1}$.By the similar arguments as given above  if we assume ${A^{\prime}}_{2a+2}=j{\prod}_{i=2}^{2a+1}{A^{\prime}}_{iQ} $, then $\{{A^{\prime}}_{2a+2}\}\bigcup{\{{A^{\prime}}_{iQ}\}}_{i=2}^{2a+1} $ is another irreducible representation of the $CA_{2a+1}$. Therefore by Lemma~ \ref{unique_last}. we get
\begin{equation}
\label{a32}
A_{2a+2}=\pm {A^{\prime}}_{2a+2}=\pm j{\prod}_{i=2}^{2a+1}{A^{\prime}}_{iQ}.
\end{equation}
Now ${\{A_{iI}\}}_{i=2}^{2a+1} \, \textrm{and} \,{\{{A^{\prime}}_{iQ}\}}_{i=2}^{2a+1} $ are two different irreducible representations of $CA_{2a}$. But from  Proposition A.5 of \cite{TiH} we know there is only one $2^{a}$ dimensional irreducible representations of $CA_{2a}$. Hence these two are equivalent representations and there  exists a \emph{special unitary} matrix $V$, i.e, $\det V =1;  \, V^{H}V=I_n$, such that,
\begin{equation*}
\label{a33}
{A^{\prime}}_{iQ}=V^{H}A_{p\left(i\right)I}V,  \, 2 \leq i \leq 2a+1
\end{equation*}
where $p(.)$ is a permutation of the set $\left\{2,3, \cdots 2a+1\right\}$. Now,
\begin{eqnarray}
\label{a34}
A_{2a+2}=\pm {A^{\prime}}_{2a+2} =V^{H}\{\pm j{\prod}_{i=2}^{2a+1}A_{iI}\}V \\
  =V^{H}\{\gamma A_{2a+1}\}V,  \textrm{  where  } \gamma \in \{+1,-1\}. \nonumber
\end{eqnarray}
Using this equation we see that
\begin{equation*}
\label{a35}
A_{1Q}=(\alpha I_n+\beta A_{2a+2}) =V^{H}(\alpha I_n+\beta \gamma A_{2a+2})V
\end{equation*}

and 

\begin{equation*}
\label{a36}
\begin{array}{l}
A_{iQ}=A_{1Q}{A^{\prime}}_{iQ}=V^{H}\left\{\alpha I_n+\beta \gamma A_{2a+2}\right\}VV^{H}A_{p(i)I}V  \\
=V^{H}\left\{\alpha A_{p\left(i\right)I}+\beta \gamma A_{2a+2}A_{p\left(i\right)I}\right\}V, ~~~~ 2 \leq i \leq 2a+1. 
\end{array}
\end{equation*}
Now, defining 
\begin{equation*}
\label{a37}
M_0=I_n;  \quad  N_0=\left(\alpha I_n+\beta \gamma A_{2a+2}\right);
\end{equation*}
\begin{equation}
\label{a38}
N_i=\left(\alpha A_{p\left(i\right)I}+\beta \gamma A_{2a+2}A_{p(i)I}\right); \quad  M_i=VA_{iI}V^{H};
\end{equation}
we have
\begin{equation}
\label{a39}
A_{iI}=V^HM_iV; \quad A_{iQ}=V^HN_iV;  \quad 1 \leq i \leq  2a+1.
\end{equation}
Notice that the matrices $M_i$ and $N_i$ are skew-Hermitian.  Now for  $2 \leq i \neq j  \leq  2a+1$, the following requirement for SSD code,
\[ A_{iI}^{H}A_{iQ}+A_{iQ}^{H}A_{iI}=0 \]
translates to, in view of \eqref{a39}, 
\begin{eqnarray*}
\label{a40}
\begin{array}{c}
V^{H}\{M_{i}^{H}N_{j}+N_{j}^{H}M_{i}\}V=0\\
\mbox{ i.e., } \quad M_{i}N_{j}=-N_{j}M_{i}, \quad \, 2 \leq i \neq j \leq  2a+1
\end{array}
\end{eqnarray*}
Now, for a specific  value of $i=k$, 
\begin{equation*}
\label{a41}
M_{k}N_{j}=-N_{j}M_{k} \quad \,  2\leq j\neq k \leq 2a+1.
\end{equation*}
On the other hand,
\begin{eqnarray*}
\begin{array}{l}
M_{i}\gamma A_{2a+2} \\
=VA_{iI}V^{H}\gamma A_{2a+2}\\
=VA_{iI}A_{2a+2}V^{H} \quad \textrm{using \eqref{a34}}\\
= -VA_{2a+2}A_{iI}V^{H}\\
= -\gamma A_{2a+2}VA_{iI}V^{H}\quad \textrm{using \eqref{a34}}\\
= -\gamma A_{2a+2}M_{i}, \textrm{  for  } \, 1\leq i \leq 2a.
\end{array}
\end{eqnarray*}
In particular, 
\begin{equation*}
\label{a42}
M_{k}\gamma A_{2a+2}=-\gamma A_{2a+2}M_{k}.
\end{equation*}
From \eqref{a32}, we see that $\{A_{2a+2}^\prime \} \bigcup\{A_{iQ}^\prime \}$ is an irreducible representation of $CA_{2a+1}$. Therefore, $\{VA_{2a+2}^\prime V^H\} \bigcup\{VA_{iQ}^\prime V^H\}$ is also an irreducible representation of $CA_{2a+1}$ since $V$ is an special unitary matrix. But,
\begin{eqnarray*}
VA_{2a+2}^\prime V^H = \pm \gamma A_{2a+2} \quad (\mbox{using  }\eqref{a34}) \\
VA_{iQ}^\prime V^H=N_i, \quad 2 \leq i \leq 2a+1, \quad (\mbox{using  } \eqref{a38}).
\end{eqnarray*}
Therefore $\{\gamma A_{2a+2}\}\bigcup {\{N_{i}\}}_{i=2}^{2a+1}$ is  representation of $CA_{2a+1}$. 

Since $\{\gamma A_{2a+2}\}\bigcup {\{N_{i}\}}_{i=2}^{2a+1}$ and  $\{\gamma A_{2a+2}, M_k \}\bigcup \{{\{N_{i}\}}_{i=2}^{2a+1}\setminus \{N_{k}\}\}$ are two irreducible representations of the generators of the $Cliff_{2a+1}$, from Lemma \ref{unique_last} we have, 
\begin{eqnarray*}
M_{k}=cN_{k}, \quad c\, \in \{+1, -1\}
\end{eqnarray*}
which leads to 
\begin{eqnarray*}
A_{kI}= V^{H}M_{k}V=c V^{H}N_{k}V=cA_{kQ} \quad \textrm{from \eqref{a38}.}
\end{eqnarray*}
contradicting  the requirement \eqref{nopathology}. Hence,  
\[ 
K \neq (2a+1).
\]
\end{proof}

\section*{Acknowledgement}
This work was partly supported by
the DRDO-IISc Program on Advanced Research in Mathematical
Engineering and by the Council of Scientific \&
Industrial Research (CSIR), India, through Research Grant (22(0365)/04/EMR-II) to B.S.~Rajan.\\
We thank X.-G.Xia for sending the preprint of \cite{WWX1}.




\end{document}